\documentclass{article}
\usepackage{amsmath,amssymb,amsfonts}

\def\op#1{\mathop{{\it\fam0} #1}\limits}
\newcommand{\id}{{\rm Id\,}}

\newcommand{\pr}{{\rm pr}}

\newcommand{\di}{{\rm dim\,}}

\newcommand{\Ker}{{\rm Ker\,}}

\newcommand{\lng}{\langle}
\newcommand{\rng}{\rangle}

\newcommand{\bite}{\begin{itemize}}
\newcommand{\eite}{\end{itemize}}
\newcommand{\benu}{\begin{enumerate}}
\newcommand{\eenu}{\end{enumerate}}
\newcommand{\bde}{\begin{description}}
\newcommand{\ede}{\end{description}}
\newcommand{\bquo}{\begin{quote}}
\newcommand{\equo}{\end{quote}}
\newcommand{\bquot}{\begin{quotation}}
\newcommand{\equot}{\end{quotation}}

\newcommand{\beq}{\begin{equation}}
\newcommand{\eeq}{\end{equation}}
\newcommand{\ben}{\begin{eqnarray}}
\newcommand{\een}{\end{eqnarray}}
\newcommand{\be}{\begin{eqnarray*}}
\newcommand{\ee}{\end{eqnarray*}}
\newcommand{\bea}{\begin{eqalph}}
\newcommand{\eea}{\end{eqalph}}

\newcommand{\cG}{{\mathfrak g}}

\newcommand{\gF}{{\mathfrak F}}

\newcommand{\cA}{{\mathcal A}}
\newcommand{\cO}{{\mathcal O}}
\newcommand{\cT}{{\mathcal T}}

\newcommand{\cV}{{\mathcal V}}
\newcommand{\cW}{{\mathcal W}}

\newcommand{\cF}{{\mathcal F}}
\newcommand{\cC}{{\mathcal C}}

\newcommand{\ccG}{{\mathcal G}}

\newcommand{\cS}{{\mathcal S}}

\newcommand{\bL}{{\mathbf L}}

\newcommand{\rA}{{\mathrm{Ann}\,}}
\newcommand{\rO}{{\mathrm{Orth}}}

\newcommand{\al}{\alpha}
\newcommand{\bt}{\beta}
\newcommand{\dl}{\delta}
\newcommand{\la}{\lambda}

\newcommand{\f}{\phi}
\newcommand{\vf}{\varphi}

\newcommand{\Om}{\Omega}
\newcommand{\m}{\mu}
\newcommand{\n}{\nu}

\newcommand{\e}{\epsilon}
\newcommand{\ve}{\varepsilon}
\newcommand{\thh}{\theta}

\newcommand{\vr}{\varrho}
\newcommand{\up}{\upsilon}
\newcommand{\vt}{\vartheta}

\newcommand{\si}{\sigma}

\newcommand{\fl}{\flat}
\newcommand{\sh}{\sharp}

\newcommand{\bT}{{\bf T}}
\newcommand{\bs}{{\bf s}}

\newcommand{\w}{\wedge}

\newcommand{\wt}{\widetilde}
\newcommand{\wh}{\widehat}

\newcommand{\dr}{\partial}

\newcommand{\ar}{\op\longrightarrow}

\let\ssection=\section
\renewcommand{\section}{\setcounter{equation}{0}\ssection}

\newcounter{eqalph}[section]
\newcounter{equationa}[section]
\newcounter{example}[section]
\newcounter{remark}[section]
\newcounter{theorem}[section]
\newcounter{proposition}[section]
\newcounter{lemma}[section]
\newcounter{corollary}[section]
\newcounter{definition}[section]

\def\theremark{\arabic{section}.\arabic{remark}}
\def\thetheorem{\arabic{section}.\arabic{theorem}}

\def\thedefinition{\arabic{section}.\arabic{definition}}

\newenvironment{proof}{\noindent {\textit{Outline of proof}:}
}{$\Box$\medskip}
\newenvironment{example}{\refstepcounter{remark}\medskip\noindent{\bf
Example \theremark:} }{$\Box$ \medskip}
\newenvironment{remark}{\refstepcounter{remark}\medskip\noindent{\bf
Remark \theremark:} }{$\Box$\medskip}
\newenvironment{theorem}{\refstepcounter{theorem}
\medskip\noindent{\sc Theorem \thetheorem}:}{$\Box$\medskip}
\newenvironment{proposition}{\refstepcounter{theorem}\medskip\noindent{\sc
Proposition \thetheorem}:}{$\Box$\medskip}
\newenvironment{lemma}{\refstepcounter{theorem}\medskip\noindent{\sc
Lemma \thetheorem}:}{ $\Box$\medskip }

\newenvironment{definition}{\refstepcounter{definition}\medskip\noindent{\sc
Definition \thedefinition}:}{$\Box$\medskip}

\newenvironment{eqalph}{\stepcounter{equation}
\setcounter{equationa}{\value{equation}} \setcounter{equation}{0}

\begin{eqnarray}}{\end{eqnarray}
\setcounter{equation}{\value{equationa}}}

\newcommand{\mar}[1]{}

\begin{document}

\hbox{}

\begin{center}

{\Large\bf

Partially superintegrable systems on Poisson manifolds}

\bigskip
\bigskip

{\sc A.KUROV$^1$, G.SARDANASHVILY$^{1,2}$}

\bigskip

$^1$ Moscow State University, Moscow, Russia

$^2$ Lepage Research Institute, Czech Republic

\bigskip
\bigskip

\end{center}

\begin{abstract}

\noindent Superintegrable (non-commutative completely integrable)
systems on a symplectic manifold conventionally are considered.
However, their definition implies a rather restrictive condition
$2n=k+m$ where $2n$ is a dimension of a symplectic manifold, $k$
is a dimension of a pointwise Lie algebra of a superintegrable
system, and $m$ is its corank. To solve this problem, we aim to
consider partially superintegrable systems on Poisson manifolds
where $k+m$ is the rank of a compatible Poisson structure. The
according extensions of the Mishchenko--Fomenko theorem on
generalized action-angle coordinates is formulated.
\end{abstract}




\section{Introduction}

The Liouville--Arnold theorem for completely integrable systems
\cite{arn,laz}, the Poincar\'e--Lyapounov--Nekhoroshev theorem for
commutative partially integrable systems \cite{gaeta,nekh94} and
the Mishchenko--Fomenko theorem (Theorem \ref{nc0}) for the
superintegrable ones \cite{bols03,fasso05,mishc} on symplectic
manifolds state the existence of action-angle coordinates around a
compact invariant submanifold of an integrable system which is a
torus $T^m$. These theorems have been extended to a general case
of invariant submanifolds which need not be compact, but are
diffeomorphic to a toroidal cylinder
\mar{g120a}\beq
\mathbb R^{m-r}\times T^r, \qquad T^r=\op\times^rS^1,
\label{g120a}
\eeq
(Theorems \ref{cmp20}, \ref{nc6} and \ref{nc0'}, respectively)
\cite{fior2,jmp03,book10,ijgmmp09a,book15,vin}.

However, Definition \ref{i0} of a superintegrable (non-commutative
completely integrable) system on a symplectic manifold is rather
restrictive. Its item (iii) requires that the matrix function
$s_{ij}$ (\ref{nc1}) must be of corank $m=2n-k$ (\ref{x1}) where
$2n$ is a dimension of a symplectic manifold and $k$ is a number
of independent generating functions of a superintegrable system.
In particular, commutative partially integrable systems on a
symplectic manifold fail to satisfy this condition (Remark
\ref{kk506}).

\begin{example} \label{x52} \mar{x52}
The Kepler problem on an $(n=2)$-dimensional configuration space
$\mathbb R^2\setminus\{0\}$ possesses three integrals of motion:
an orbital momentum $M_{12}$ and two components of the Rung--Lenz
vector $A^a$ \cite{ijgmmp09a,book15}. They constitute two
different superintegrable systems: (i) with a Lie algebra $so(3)$
on a domain $U_-$ of a phase space $\mathbb R^4\setminus\{0\}$ of
negative energy, and (ii) with a Lie algebra $so(2,1)$ on a domain
$U_+\subset\mathbb R^4\setminus\{0\}$ of positive energy. However,
if $n>2$, a number of integrals of motion $(M_{ij}, A^i)$,
$i=1,\ldots, n$, of the Kepler system is more than $2n$, and they
fail to form any superintegrable system. In this case, one however
can consider a partially superintegrable system with the
generating functions $(M_{12}, A^1, A^2)$ because orbits of a
motion of a Kepler system is well known to lie in a plane.
\end{example}

To avoid the restriction $m=2n-k$ (\ref{x1}), we aim to consider
partially superintegrable systems on Poisson manifolds (Definition
\ref{x11}). The following are their examples.

\begin{example} \label{x6} \mar{x6}
Let $F=(F_1,\ldots,F_k)$ be a superintegrable system of corank
$2n-k$ on a $2n$-dimensional connected symplectic manifold
$(Z,\Om)$. Let $X$ be an $r$-dimensional manifold regarded as a
Poisson manifold with a zero Poisson structure (Example
\ref{w206}). A manifold product $Z\times Q$ can be endowed with a
product Poisson structure $w$ of rank $2n$ (Example \ref{spr862}).
Then the pull-back of functions $F_i$ onto $Z\times X$ exemplifies
a partially superintegrable system of corank $m=2n-k$ where $2n$
is a rank of a Poisson structure on a $(2n+r)$-dimensional Poisson
manifold (Definition \ref{x11}).
\end{example}

\begin{example}  \label{x7} \mar{x7}
Commutative partially integrable systems on Poisson manifolds
(Definition \ref{cc1}) exemplify partially superintegrable
systems.
\end{example}

A key point is that invariant submanifolds of a superintegrable
system are integral manifolds of a certain commutative
 partially integrable system on a symplectic manifold (Remark
\ref{x2}). As a consequence, the proof of above mentioned
generalized Mishchenko--Fomenko theorem (Theorem \ref{nc0'}) for
superintegrable systems is reduced to generalized
Poincar\'e--Lyapounov--Nekhoroshev Theorem \ref{nc6} for
commutative partially integrable systems on a symplectic manifold
(Section 2). Therefore, we start our investigation of partially
superintegrable systems with commutative partially integrable
systems on Poisson manifolds (Section 4) \cite{fior,jmp03}.

Our goal is that, in a case of partially integrable systems on
Poisson manifolds, the above mentioned restriction condition
(\ref{x1}) comes to a form $k+m=r$ where $r$ is the rank of a
Poisson structure on a manifold $Z$, but not a dimension of $Z$
(Lemma \ref{x50}).

The extended Mishchenko--Fomenko theorem on generalized
action-angle coordinates in the case of symplectic superintegrable
systems (Theorem \ref{nc0'}) is extended to partially
superintegrable systems on Poisson manifolds (Theorem \ref{x51}).
Its proof also is reduced to Theorems \ref{bi92} and \ref{bi100}
for commutative partially integrable systems on Poisson manifolds.

\section{Superintegrable systems on symplectic manifolds}

Throughout the work, all functions and maps are smooth, and
manifolds are finite-dimensional smooth real and paracompact.

\begin{definition} \label{i0} \mar{i0}
Let $(Z,\Om)$ be a $2n$-dimensional connected symplectic manifold,
and let $(C^\infty(Z), \{,\})$ be a Poisson algebra of smooth real
functions on $Z$. A subset
\mar{i00}\beq
F=(F_1,\ldots,F_k), \qquad n\leq k<2n, \label{i00}
\eeq
of a Poisson algebra $C^\infty(Z)$ is called a superintegrable
system if the following conditions hold.

(i) All the functions $F_i$ (called the  generating functions of a
superintegrable system) are independent, i.e., a $k$-form
$\op\w^kdF_i$ nowhere vanishes on $Z$. It follows that a
surjection
\mar{nc4}\beq
\wh F:Z\to N=\op\times_iF_i(Z)\subset \mathbb R^k \label{nc4}
\eeq
is a submersion, i.e., a fibred manifold over a domain
(contractible open subset) $N\subset\mathbb R^k$ endowed with the
coordinates $(x_i)$ such that $x_i\circ \wh F=F_i$. Fibres of the
fibred manifold $\wh F$ (\ref{nc4}) are called the invariant
submanifolds of a superintegrable system.

(ii) There exist smooth real functions $s_{ij}$ on $N$ such that
\mar{nc1}\beq
\{F_i,F_j\}= s_{ij}\circ \wh F, \qquad i,j=1,\ldots, k.
\label{nc1}
\eeq

(iii) The matrix function $\bs$ with the entries $s_{ij}$
(\ref{nc1}) is of constant corank
\mar{x1}\beq
m=2n-k, \qquad 2n=\di Z, \qquad k=\di N, \label{x1}
\eeq
at all points of $N$.
\end{definition}

If $k>n$, the matrix $\bs$ is necessarily non-zero. If $k=n$, then
$\bs=0$, and we are in the case of completely integrable systems
as follows.

\begin{definition} \label{cmp21} \mar{cmp21} The subset $(F_1,\ldots,F_n)$ (\ref{i00})
of a Poisson algebra $C^\infty(Z)$ on a symplectic manifold
$(Z,\Om)$ is called the completely integrable system if $F_i$ are
independent functions in involution.
\end{definition}

Therefore, superintegrable systems sometimes are called
non-commutative completely integrable systems. However, this
notion differs from that in \cite{laur}.

\begin{remark} \label{kk506} \mar{kk506} A
family $\{S_1,\ldots, S_m\}$ of $m\leq n$ independent smooth real
functions in involution on a symplectic manifold $(Z,\Om)$ is
called the (commutative) partially integrable system. It should be
emphasized that a commutative partially integrable system on a
symplectic manifold fails to be a particular superintegrable
system because the condition (\ref{x1}) in item (iii) of
Definition \ref{i0} is not satisfied, unless $m=n$ and it is a
completely integrable system.
\end{remark}

The following two assertions clarify the structure of
superintegrable systems \cite{fasso05,fior2,book10,book15}.

\begin{lemma} \label{nc7} \mar{nc7} Given a symplectic manifold $(Z,\Om)$,
let $\pi:Z\to N$ be a fibred manifold such that, for any two
functions $f$, $f'$ constant on fibres of $\pi$, their Poisson
bracket $\{f,f'\}$ also is well. By virtue of Theorem \ref{p11.3},
$N$ is provided with a unique coinduced Poisson structure
$\{,\}_N$ such that $\pi$ is a Poisson morphism.
\end{lemma}

Since any function constant on fibres of $\pi$ is a pull-back of
some function on $N$, the superintegrable system (\ref{i00}) with
$\pi=\wh F$ satisfies the condition of Lemma \ref{nc7} due to item
(i) of Definition \ref{i0}. Thus, a base $N$ of the fibred
manifold $\wh F$ (\ref{nc4}) is endowed with a coinduced Poisson
structure of corank $m$. With respect to coordinates $x_i$ in item
(i) of Definition \ref{i0}, its bivector field reads
\mar{cmp1}\beq
w=s_{ij}(x_k)\dr^i\w\dr^j. \label{cmp1}
\eeq

\begin{lemma} \label{nc8} \mar{nc8} Given a fibred manifold $\pi:Z\to N$ with connected fibres in
Lemma \ref{nc7}, the following conditions are equivalent
\cite{fasso05,libe}:

(i) the corank of the coinduced Poisson structure $\{,\}_N$ on $N$
equals $m=\di Z-\di N$,

(ii) the fibres of $\pi$ are isotropic,

(iii) the fibres of $\pi$ are  maximal integral manifolds of the
involutive distribution spanned by the Hamiltonian vector fields
of the pull-back $\pi^*C$ of Casimir functions $C$ of a coinduced
Poisson structure on $N$.
\end{lemma}

It is readily observed that the fibred manifold $\wh F$
(\ref{nc4}) obeys condition (i) of Lemma \ref{nc8} due to item
(iii) of Definition \ref{i0}, namely, $k-m= 2(k-n)$.

\begin{remark} \label{x2} \mar{x2}
The pull-back $\pi^*C$ of Casimir functions in item (iii) of Lemma
\ref{nc8} are in involution with all functions on $Z$ which are
constant on fibres of $\pi$. Let $N$ admit a family of $m$
independent Casimir functions $(C_\la)$. Then their pull-back
$\pi^*C_\la$ constitute a commutative partially integrable system
on a symplectic manifold $(Z,\Om)$ (Remark \ref{kk506}). In this
case it follows from item (iii) of Lemma \ref{nc8} that invariant
submanifolds of a superintegrable system are integral manifolds of
this commutative partially integrable system. As a consequence,
the proof of generalized Mishchenko--Fomenko Theorem \ref{nc0'} is
reduced to the generalized  Poincar\'e--Lyapounov--Nekhoroshev
theorem (Theorem \ref{nc6}) for commutative partially integrable
systems on a symplectic manifold.
\end{remark}

\begin{remark} \label{cmp8} \mar{cmp8} In many applications, condition (i) of Definition
\ref{i0} fails to hold. It can be replaced with that a subset
$Z_R\subset Z$ of regular points (where $\op\w^kdF_i\neq 0$) is
open and dense. Let $M$ be an invariant submanifold through a
regular point $z\in Z_R\subset Z$. Then it is regular, i.e.,
$M\subset Z_R$. Let $M$ admit a regular open saturated
neighborhood $U_M$ (i.e., a fibre of $\wh F$ through a point of
$U_M$ belongs to $U_M$). For instance, any compact invariant
submanifold $M$ has such a neighborhood $U_M$. The restriction of
functions $F_i$ to $U_M$ defines a superintegrable system on $U_M$
which obeys Definition \ref{i0}. In this case, one says that a
superintegrable system is considered around its invariant
submanifold $M$. We refer to \cite{ijgmmp09a,book15} for a global
analysis of superintegrable systems.
\end{remark}

Given a superintegrable system in accordance with Definition
\ref{i0}, the above mentioned generalization of the Mishchenko --
Fomenko theorem to non-compact invariant submanifolds states the
following \cite{fior2,book10,ijgmmp09a,book15}.

\begin{theorem} \label{nc0'} \mar{nc0'} Let the Hamiltonian vector fields $\vt_i$ (\ref{z100}) of the
generating functions $F_i$ be complete, and let fibres of the
fibred manifold $\wh F$ (\ref{nc4}) be connected and mutually
diffeomorphic. Then the following hold.

(I) The fibres of the fibred manifold $\wh F$ (\ref{nc4}) are
diffeomorphic to the toroidal cylinder (\ref{g120a}) where
$m=2n-k$.

(II) Given a fibre $M$ of the fibration $\wh F$ (\ref{nc4}), there
exists its open saturated neighborhood $U_M$ which is a trivial
principal bundle
\mar{kk600}\beq
U_M=N_M\times (\mathbb R^{m-r}\times T^r)\ar^{\wh F} N_M
\label{kk600}
\eeq
with the structure additive group (\ref{g120a}).

(III) A neighborhood $U_M$ is provided with the bundle
(generalized action-angle) coordinates $(I_\la,q^A, y^\la)$,
$\la=1,\ldots, m$, $A=1,\ldots,2(n-m)$, such that: (i) the
generalized angle coordinates $(y^\la)$ are coordinates on a
toroidal cylinder, i.e., fibre coordinates on the fibre bundle
(\ref{kk600}); (ii) the $(I_\la,q^A)$ are coordinates on its base
$N_M$ where the action coordinates $(I_\la)$ are values of
independent Casimir functions of the coinduced Poisson structure
$\{,\}_N$ on $N_M$; and (iii) a symplectic form $\Om$ on $U_M$
reads
\be
\Om= dI_\la\w dy^\la + \Om_{AB}(I_\bt,q^C)dq^A\w dq^B.
\ee
\end{theorem}

\begin{proof}
It follows from item (iii) of Lemma \ref{nc8} that every fibre $M$
of the fibred manifold (\ref{nc4}) is a maximal integral manifold
of an involutive distribution spanned by the Hamiltonian vector
fields $\up_\la$ of the pull-back $\wh F^*C_\la$ of $m$
independent Casimir functions $\{C_1,\ldots, C_m\}$ of the
coinduced Poisson structure $\{,\}_N$ (\ref{cmp1}) on an open
neighborhood $N_M$ of a point $\wh F(M)\in N$. Let us put $U_M=\wh
F^{-1}(N_M)$. It is an open saturated neighborhood of $M$.
Consequently, invariant submanifolds of a superintegrable system
(\ref{i00}) on $U_M$ are integral manifolds of a commutative
partially integrable system
\mar{cmp4}\beq
S=(\wh F^*C_1, \ldots, \wh F^*C_m), \qquad 0<m\leq n, \label{cmp4}
\eeq
on a symplectic manifold $(U_M,\Om)$ (Remark \ref{x2}). Therefore,
statements (I) -- (III) of Theorem \ref{nc0'} are the corollaries
of Theorem \ref{nc6} below. Its condition (i) is satisfied as
follows. Let $M'$ be an arbitrary fibre of the fibred manifold
$\wh F:U_M\to N_M$ (\ref{nc4}). Since
\be
\wh F^*C_\la(z)= (C_\la\circ \wh F)(z)= C_\la(F_i(z)),\qquad z\in
M',
\ee
the Hamiltonian vector fields $\up_\la$ on $M'$ are $\mathbb
R$-linear combinations of Hamiltonian vector fields $\vt_i$ of
generating functions $F_i$. It follows that $\up_\la$ on $M'$ are
elements of a finite-dimensional real Lie algebra of vector fields
on $M'$ generated by vector fields $\vt_i$. Since vector fields
$\vt_i$ are complete, the vector fields $\up_\la$ on $M'$ also are
well (Remark \ref{zz95} below). Consequently, these vector fields
are complete on $U_M$ because they are vertical vector fields on
$U_M\to N$.
\end{proof}

\begin{remark} \label{zz95} \mar{zz95} If complete vector fields on a
smooth manifold constitute a basis for a finite-dimensional real
Lie algebra, any element of this Lie algebra is complete
\cite{palais}.
\end{remark}

\begin{remark} \label{zz90} The condition of the completeness of Hamiltonian
vector fields of generating functions $F_i$ in Theorem \ref{nc0'}
is rather restrictive. One can replace it with that the
Hamiltonian vector fields $\up$ of the pull-back onto $Z$ of
Casimir functions on $N$ are complete.
\end{remark}

If the conditions of Theorem \ref{nc0'} are replaced with that
fibres of the fibred manifold $\wh F$ (\ref{nc4}) are compact and
connected, this theorem restarts the Mishchenko--Fomenko one as
follows \cite{bols03,fasso05,mishc}.

\begin{theorem} \label{nc0} \mar{nc0}
Let fibres of the fibred manifold $\wh F$ (\ref{nc4}) be compact
and connected. Then they are diffeomorphic to a torus $T^m$, and
statements (II) -- (III) of Theorem \ref{nc0'} hold.
\end{theorem}

\begin{remark}
In Theorem \ref{nc0}, the Hamiltonian vector fields $\up_\la$ are
complete because fibres of the fibred manifold $\wh F$ (\ref{nc4})
are compact. As well known, any vector field on a compact manifold
is complete.
\end{remark}

If $F$ (\ref{i00}) is a completely integrable system, the
coinduced Poisson structure on $N$ equals zero, and the generating
functions $F_i$ are the pull-back of $n$ independent functions on
$N$. Then Theorems \ref{nc0} and \ref{nc0'} come to the well-known
Liouville--Arnold theorem \cite{arn,laz} and its generalization
(Theorem \ref{cmp20} below) to the case of non-compact invariant
submanifolds \cite{fior,book10,book15}, respectively.

\begin{theorem} \label{cmp20} \mar{cmp20} Given a completely integrable system
$F$ in accordance with Definition \ref{cmp21}, let the Hamiltonian
vector fields $\vt_i$ of functions $F_i$ be complete, and let
fibres of the fibred manifold $\wh F$ (\ref{nc4}) be connected and
mutually diffeomorphic. Then items (I) and (II) of Theorem
\ref{nc0'} hold, and its item (III) is replaced with the following
one.

(III') The neighborhood $U_M$ (\ref{kk600}) where $m=n$ is
provided with the bundle (generalized action-angle) coordinates
$(I_\la,y^\la)$, $\la=1,\ldots, n$, such that the angle
coordinates $(y^\la)$ are coordinates on a toroidal cylinder, and
the symplectic form $\Om$ on $U_M$ reads
\mar{cmp66}\beq
\Om= dI_\la\w dy^\la. \label{cmp66}
\eeq
\end{theorem}

In a general setting, one considers commutative partially
integrable systems on a symplectic manifold (Remark \ref{kk506}).
The Poincar\'e--Lyapounov--Nekhoroshev theorem
\cite{gaeta,gaeta03,nekh94} generalizes the Liouville--Arnold one
to a commutative partially integrable system with compact
invariant submanifolds. Forthcoming Theorem \ref{nc6} is concerned
with a generic commutative partially integrable system on a
symplectic manifold \cite{jmp03,ijgmmp09a,book15}.

Given a commutative partially integrable system $S=\{S_1,\ldots,
S_m\}$ of $m\leq n$ independent smooth real functions in
involution on a symplectic manifold $(Z,\Om)$, we have a fibred
manifold
\mar{g106}\beq
\wh S:Z\to X=\op\times^m S_\la(Z)\subset\mathbb R^m \label{g106}
\eeq
over a domain $X\subset \mathbb R^m$. We agree to call its fibres
the invariant submanifolds of a commutative partially integrable
system though it is not a superintegrable system (Remark
\ref{kk506}),  unless $m=n$ and it is a completely integrable
system.

Hamiltonian vector fields $v_\la$ of generating functions $S_\la$
are mutually commutative and independent. Consequently, they span
an $m$-dimensional involutive distribution on $Z$ whose maximal
integral manifolds constitute an isotropic foliation $\cS$ of $Z$.
Its leaves are called the integral manifolds of a commutative
partially integrable system. Because functions $S_\la$ are
constant on leaves of this foliation, each fibre of the fibred
manifold $\wh S$ (\ref{g106}) is foliated by the leaves of a
foliation $\cS$.

If $m=n$, we are in the case of a completely integrable system,
and its integral manifolds (i.e., leaves of $\cS$) are connected
components of its invariant submanifolds (i.e., fibres of the
fibred manifold $\wh S$ (\ref{g106})).

\begin{theorem} \label{nc6} \mar{nc6}
Let a commutative partially integrable system
$S=\{S_1,\ldots,S_m\}$ on a symplectic manifold $(Z,\Om)$ satisfy
the following conditions.

(i) The Hamiltonian vector fields $v_\la$ of $S_\la$ are complete.

(ii) The foliation $\cS$ is a fibred manifold $\pi_\cS:Z\to N$
whose fibres are mutually diffeomorphic and, being integral
manifolds, are connected.

\noindent Then the following hold
\cite{jmp03,book10,ijgmmp09a,book15}.

(I) The fibres of $\cS$ are diffeomorphic to the toroidal cylinder
(\ref{g120a}).

(II) Given a fibre $M$ of $\cS$, there exists its open saturated
neighborhood $U_M$ such that the restriction of $\pi_\cS$ to $U_M$
is a trivial principal bundle with the structure additive group
(\ref{g120a}), and we have a composite fibre bundle
\mar{x30}\beq
\wh S: U_M\ar \pi_\cS(U_M)\ar \wh S(U_M)\subset\mathbb R^m.
\label{x30}
\eeq

(III) A neighborhood $U_M$ is provided with the bundle
(generalized action-angle) coordinates
\mar{x24}\beq
 (I_\la,q^A,y^\la)\to
(I_\la,q^A)\to (I_\la), \quad \la=1,\ldots,m, \quad A=1,\ldots,
2(n-m), \label{x24}
\eeq
such that: (i) the action coordinates $(I_\la)$ on $\wh S(U_M)$
are expressed into the values of generating functions $S_\la$;
(ii) the angle coordinates $(y^\la)$ are coordinates on the
toroidal cylinder (\ref{g120a}); and (iii) a symplectic form $\Om$
on $U_M$ reads
\be
\Om= dI_\la\w dy^\la +  \Om_{AB}(I_\bt,q^C) dq^A\w dq^B.
\ee
\end{theorem}

\section{Lie algebra superintegrable systems}

Following the original Mishchenko--Fomenko theorem
\cite{bols03,fasso05,mishc}, let us consider superintegrable
systems whose generating functions $F=\{F_1,\ldots,F_k\}$ form a
$k$-dimensional real Lie algebra $\cG$ of corank $m=2n-k$ with the
commutation relations
\mar{zz60}\beq
\{F_i,F_j\}= c_{ij}^h F_h, \qquad c_{ij}^h={\rm const.}
\label{zz60}
\eeq
We agree to call them the Lie algebra superintegrable systems. In
this case, the fibration $\wh F$ (\ref{nc4}) is a momentum mapping
of $Z$ onto a domain $N$ of the Lie coalgebra $\cG^*$ (Section
6.2) which is provided with the coordinates $x_i$ in item (i) of
Definition \ref{i0} \cite{book10,guil,book15}. Accordingly, the
coinduced Poisson structure $\{,\}_N$ on $N$ coincides with the
canonical Lie--Poisson structure on $\cG^*$, and it is given by a
Poisson bivector field
\be
w=\frac12 c_{ij}^h x_h\dr^i\w\dr^j.
\ee

In view of the relations (\ref{m81}), Hamiltonian vector fields
$\vt_i$ of generating functions $F_i$ of the Lie algebra
superintegrable system (\ref{zz60}) make up a real Lie algebra
$\cG$ with the commutation relations
\mar{x8}\beq
[\vt_i,\vt_j]= c_{ij}^h \vt_h. \label{x8}
\eeq
Since the morphism $w^\sh$ (\ref{m51}) is of maximal rank,
Hamiltonian vector fields $\vt_i$ are independent, i.e.,
$\op\w^k\vt_i\neq 0$.

Following the conditions of Theorem \ref{nc0'}, let us assume that
Hamiltonian vector fields $\vt_i$ are complete. In accordance with
the above mentioned theorem \cite{onish,palais}, they define a
Hamiltonian action on $Z$ of a simply connected Lie group $G$
whose Lie algebra is isomorphic to $\cG$. Since vector fields
$\vt_i$ are independent, the action of $G$ on $Z$ is locally free,
i.e., isotropy groups of points of $U$ are discrete subgroups of
$G$. Orbits of $G$ coincide with $k$-dimensional maximal integral
manifolds of a regular distribution on $Z$ spanned by Hamiltonian
vector fields $\vt_i$ \cite{susm}. They constitute a foliation
$\cF$ of $Z$. Then the fibration $\wh F$ (\ref{nc4}) sends its
leaves $\cF_z$ through points $z\in Z$ onto the orbits $\ccG_{\wh
F(z)}$ of the coadjoint action (\ref{z211}) of $G$ on $\cG^*$,
which coincide with the canonical symplectic foliation $\ccG$ of
$\cG^*$. Conversely, $\cF_z=\wh F^{-1}(\ccG_{\wh F(z)})$ in
accordance with item (iii) of Lemma \ref{nc8}.

It should be noted note that Casimir functions $C\in \cC(\cG^*)$
of the Lie--Poisson structure on $\cG^*$ are exactly the coadjoint
invariant functions on $\cG^*$. They are constant on orbits of the
coadjoint action of $G$ on $\cG^*$. Consequently, their pull-back
$\wh F^*C$ are constant on leaves of a foliation $\cF$. Therefore,
the real Lie algebra $\cG$ (\ref{x8}) is extended to a Lie algebra
over a subring $\cC =\wh F^*\cC(N)\subset C^\infty(Z)$ of the
pull-back $\wh FC$ of Casimir functions on $N\subset \cG^*$.

Now let us assume that a foliation $\cF$ is a fibred manifold
$\pi_\cF: Z\to \pi_\cF(Z)$ whose fibres are mutually
diffeomorphic. This implies that a symplectic foliation $\ccG$ of
$\wh F(Z)$ also is a fibred manifold $\pi_\ccG:\wh F(Z)\to
\pi_\cF(Z)$. Thus, we have a composite fibred manifold
\mar{x20}\beq
\pi_\cF=\pi_\ccG\circ\wh F: Z\ar \wh F(Z)\ar \pi_\cF(Z),
\label{x20}
\eeq
so that a fibration $\wh F$ obeys the conditions of Theorem
\ref{nc0'}. It follows that, given a leaf $V$ of $\cF$, there
exists its open saturated neighborhood $U_V$ such that $\wh
F(U_V)\subset \cG^*$ is provided with some family of $m$
independent Casimir functions $C=(C_1,\ldots,C_m)$ and, restricted
to $U_V$, the composite fibred manifold $\pi_\cF$ (\ref{x20})
becomes a composite bundle
\mar{x21}\ben
&& \pi_\cF=\wh C\circ\wh F: U_V\ar \wh F(U_V)\ar \op\times^m
C_\la(\wh F(U_V))= \label{x21}\\
&& \qquad \op\times^m \wh F^*C_\la(U_V)=\pi_\cF(U_V) \nonumber
\een
in toroidal cylinders. In accordance with Remark \ref{x2}, fibres
of $\wh F$ are integral manifolds of a commutative partially
integrable system $S=(S_\la)$ on $Z$ of the pull-back $S_\la=\wh
F^*C_\la$ of Casimir functions $C_\la$. Then the composite bundle
(\ref{x21}) takes the form (\ref{x30}). Accordingly, it is endowed
with the bundle (generalized action-angle) coordinates
\mar{x23}\beq
(I_\la,q^A,y^\la)\to (I_\la,q^A)\to (I_\la), \quad \la=1,\ldots,m,
\quad A=1,\ldots, k-m, \label{x23}
\eeq
which are generalized action-angle coordinates (\ref{x24}) in
Theorem \ref{nc6} when $S_\la$ are the pull-back $\wh F^*C_\la$ of
Casimir functions $C_\la$. Note that the latter in turn are the
pull-back $C_\la=\pi_\ccG^*\Phi_\la$ of some functions $\Phi_\la$
on a base $\pi_\cF(U_V)$ of the composite bundle (\ref{x21}).

\section{Commutative partially integrable systems on Poisson manifolds}

As was mentioned above invariant submanifolds of a superintegrable
system are integral manifolds of a certain commutative partially
integrable system on a symplectic manifold (Remark \ref{x2}).
Therefore, we start our analysis of partially superintegrable
systems with commutative partially integrable systems on Poisson
manifolds in Example \ref{x7}
\cite{fior,jmp03,kurov,book10,book15}.

A key point is that a commutative partially integrable system
admits different compatible Poisson structures (Theorem
\ref{bi92}). Treating commutative partially integrable systems, we
therefore are based on a wider notion of the commutative dynamical
algebra \cite{jmp03}.

Let we have $m$ mutually commutative vector fields $\{\vt_\la\}$
on a connected smooth manifold $Z$ which are functionally
independent (i.e., $\op\w^m\vt_\la\neq 0$) everywhere on $Z$. We
denote by $\cC\subset C^\infty(Z)$ a $\mathbb R$-subring of smooth
real functions $f$ on $Z$ whose derivations $\vt_\la\rfloor df$
along $\vt_i$ vanish for all $\vt_\la$. Let $\cA$ be an
$m$-dimensional Lie $\cC$-algebra generated by the vector fields
$\{\vt_\la\}$.

\begin{definition} \label{kk1} \mar{kk1}
We agree to call $\cA$ the commutative dynamical algebra.
\end{definition}

For instance, given a commutative partially integrable system $S$
on a symplectic manifold (Remark \ref{kk506}), the Hamiltonian
vector fields of its generating functions constitute a commutative
dynamical algebra in accordance with Definition \ref{kk1}.

In a general setting, let us now consider a commutative dynamical
algebra on a Poisson manifold.

\begin{definition} \label{cc1} \mar{cc1} Let $(Z,w)$ be a
(regular) Poisson manifold (Section 6.4) and $\cA$ an
$m$-dimensional commutative dynamical algebra on $Z$. A triple
$(Z,\cA,w)$ is said to be a commutative partially integrable
system if the following hold.

(a) The generators $\vt_\la$ are Hamiltonian vector fields of some
independent functions $S_\la\in \cC$ on $Z$.

(b) All elements of $\cC\subset C^\infty(Z)$ are mutually  in
involution, i.e., their Poisson brackets $\{f,f'\}_w$, $f,f'\in
\cC$, equal zero.
\end{definition}

It follows at once from this definition that the Poisson structure
$w$ is at least of rank $2m$, and that $\cC$ is a commutative
Poisson algebra. We call the functions $S_\la$ in item (a) of
Definition \ref{cc1} the  generating functions of a commutative
partially integrable system, which is uniquely defined by a family
$(S_1,\ldots,S_m)$ of these functions.

\begin{definition} \label{kk2} \mar{kk2} We say that a Poisson
structure in Definition \ref{kk1} is compatible.
\end{definition}

One can show (Theorem \ref{bi92}) that a compatible Poisson
structure is of rank $2m$.

\begin{remark} \label{000} \mar{000}
If $2m=\di Z$ in Definition \ref{cc1}, we have a  completely
integrable system on a symplectic manifold $Z$ (Definition
\ref{cmp21}). However, a commutative partially integrable system
on a symplectic manifold in Remark \ref{kk506} fails to be well in
accordance with Definition \ref{cc1} because it does not satisfy
item (b) of this Definition if $m<n$.
\end{remark}

If $2m<\di Z$, there exist different compatible Poisson structures
on $Z$ which bring a commutative dynamical algebra $\cA$ into a
commutative partially integrable system.

Forthcoming Theorem \ref{bi0} shows that, under certain
conditions, there exists a compatible Poisson structure on an open
neighborhood of an invariant submanifold $M$ of a commutative
dynamical algebra.  Theorems \ref{bi92} -- \ref{bi72}  describe
all these Poisson structures around an invariant submanifold
$M\subset Z$ of $\cA$ \cite{jmp03}. Given a commutative partially
integrable system $(w,\cA)$ in Theorem \ref{bi92}, the bivector
field $w$ (\ref{bi42}) can be brought into the canonical form
(\ref{kk500}) with respect to generalized action-angle coordinates
in Theorem \ref{bi100}. This theorem extends the above-mentioned
Liouville--Arnold and Poincar\'e--Lyapounov--Nekhoroshev theorems
to the case of a Poisson structure and a non-compact invariant
submanifold \cite{jmp03,book10,book15}.

Given a commutative dynamical algebra $\cA$ on a manifold $Z$, let
$G$ be the group of local diffeomorphisms of $Z$ generated by the
flows of its elements. The orbits of $G$ are maximal invariant
submanifolds of $\cA$ (we follow the terminology of \cite{susm}).
Tangent spaces to these submanifolds form a (regular) distribution
$\cV\subset TZ$ whose maximal integral manifolds coincide with
orbits of $G$. Being involutive, this distribution yields a
foliation $\cS$ of $Z$ (Section 6.1).

\begin{theorem} \label{bi0} \mar{bi0}
Let $\cA$ be a commutative dynamical algebra, $M$ its invariant
submanifold, and $U$ a saturated open neighborhood of $M$ (Remark
\ref{cmp8}). Let us suppose that:

(i) vector fields $\vt_\la$ on $U$ are complete,

(ii) a foliation $\cS$ of $U$ is a fibred manifold $\pi_\cS$ with
mutually diffeomorphic fibres.

\noindent Then the following hold \cite{jmp03,book10,book15}.

(I) Leaves of $\cS$ are diffeomorphic to the toroidal cylinder
(\ref{g120a}).

(II) There exists an open saturated neighborhood of $M$, say $U$
again, which is a trivial principal bundle
\mar{z10'}\beq
U=N\times(\mathbb R^{m-r}\times T^r)\ar^{\pi_\cS} N \label{z10'}
\eeq
with the structure additive group (\ref{g120a}) over a domain
$N\subset \mathbb R^{\di Z-m}$.

(III) If $2m\leq\di Z$, there exists a Poisson structure of rank
$2m$ on $U$ such that $\cA$ is a commutative partially integrable
system in accordance with Definition \ref{cc1}.
\end{theorem}

\begin{proof} (I) Since $m$-dimensional leaves of the foliation $\cF$ admit $m$
complete independent vector fields, they are locally affine
manifolds diffeomorphic to the toroidal cylinder (\ref{g120a}).

(II) Since a foliation $\cF$ of $U$ is a fibred manifold by virtue
of item (ii), one can always choose an open fibred neighborhood of
its fibre $M$, say $U$ again, over a domain $N$ such that this
fibred manifold
\mar{d20}\beq
\pi:U\to N \label{d20}
\eeq
admits a section $\si$. In accordance with the above mentioned
theorem \cite{onish,palais}, complete Hamiltonian vector fields
$\vt_\la$ define an action of a simply connected Lie group $G$ on
$Z$. Because vector fields $\vt_\la$ are mutually commutative, it
is an additive group $\mathbb R^m$ whose group space is
coordinated by parameters $s^\la$ of the flows with respect to the
basis $\{e_\la=\vt_\la\}$ for its Lie algebra. Orbits of a group
$\mathbb R^m$ in $U\subset Z$ coincide with fibres of the fibred
manifold (\ref{d20}). Since vector fields $\vt_\la$ are
independent everywhere on $U$, the action of $\mathbb R^m$ on $U$
is locally free, i.e., isotropy groups of points of $U$ are
discrete subgroups of a group $\mathbb R^m$. Given a point $x\in
N$, the action of $\mathbb R^m$ on a fibre $M_x=\pi^{-1}(x)$
factorizes as
\mar{d4}\beq
\mathbb R^m\times M_x\to G_x\times M_x\to M_x \label{d4}
\eeq
through the free transitive action on $M_x$ of the factor group
$G_x=\mathbb R^m/K_x$, where $K_x$ is the isotropy group of an
arbitrary point of $M_x$. It is the same group for all points of
$M_x$ because $\mathbb R^m$ is a commutative group. Clearly, $M_x$
is diffeomorphic to a group space of $G_x$. Since fibres $M_x$ are
mutually diffeomorphic, all isotropy groups $K_x$ are isomorphic
to the group $\mathbb Z_r$ for some fixed $0\leq r\leq m$.
Accordingly, the groups $G_x$ are isomorphic to the additive group
(\ref{g120a}). Let us bring the fibred manifold (\ref{d20}) into a
principal bundle with a structure group $G_0$, where we denote
$\{0\}=\pi(M)$. For this purpose, let us determine isomorphisms
$\rho_x: G_0\to G_x$ of a group $G_0$ to groups $G_x$, $x\in N$.
Then a desired fibrewise action of $G_0$ on $U$ is defined by the
law
\mar{d5}\beq
G_0\times M_x\to\rho_x(G_0)\times M_x\to M_x. \label{d5}
\eeq
Generators of each isotropy subgroup $K_x$ of $\mathbb R^m$ are
given by $r$ linearly independent vectors of the group space
$\mathbb R^m$. One can show that there exist ordered collections
of generators $(v_1(x),\ldots,v_r(x))$ of the groups $K_x$ such
that $x\to v_i(x)$ are smooth $\mathbb R^m$-valued fields on $N$.
Indeed, given a vector $v_i(0)$ and a section $\si$ of the fibred
manifold (\ref{d20}), each field $v_i(x)=(s^\al_i(x))$ is a unique
smooth solution of an equation
\be
g(s_i^\al)\si(x)=\si(x), \qquad  (s_i^\al(0))=v_i(0),
\ee
on an open neighborhood of $\{0\}$. Let us consider the
decomposition
\be
v_i(0)=B_i^a(0) e_a + C_i^j(0) e_j, \qquad a=1,\ldots,m-r, \qquad
j=1,\ldots, r,
\ee
where $C_i^j(0)$ is a non-degenerate matrix. Since the fields
$v_i(x)$ are smooth, there exists an open neighborhood of $\{0\}$,
say $N$ again, where the matrices $C_i^j(x)$ are non-degenerate.
Then
\mar{d6}\beq
A(x)=\left(
\begin{array}{ccc}
\id & \qquad & (B(x)-B(0))C^{-1}(0) \\
0 & & C(x)C^{-1}(0)
\end{array}
\right) \label{d6}
\eeq
is a unique linear endomorphism
\be
(e_a,e_i)\to (e_a,e_j)A(x)
\ee
of a vector space $\mathbb R^m$ which transforms a frame
$\{v_\la(0)\}=\{e_a,v_i(0)\}$ into a frame
$\{v_\la(x)\}=\{e_a,\vt_i(x)\}$, i.e.,
\be
v_i(x)=B_i^a(x) e_a + C_i^j(x) e_j=B_i^a(0) e_a + C_i^j(0)
[A_j^b(x)e_b +A_j^k(x)e_k].
\ee
Since $A(x)$ (\ref{d6}) also is an automorphism of a group
$\mathbb R^m$ sending $K_0$ onto $K_x$, we obtain a desired
isomorphism $\rho_x$ of a group $G_0$ to a group $G_x$. Let an
element $g$ of a group $G_0$ be the coset of an element $g(s^\la)$
of a group $\mathbb R^m$. Then it acts on $M_x$ by the rule
(\ref{d5}) just as the element $g((A_x^{-1})^\la_\bt s^\bt)$ of a
group $\mathbb R^m$ does. Since entries of the matrix $A$
(\ref{d6}) are smooth functions on $N$, this action of a group
$G_0$ on $U$ is smooth. It is free, and $U/G_0=N$. Then the fibred
manifold (\ref{d20}) is a trivial principal bundle with a
structure group $G_0$. Given a section $\si$ of this principal
bundle, its trivialization $U=N\times G_0$ is defined by assigning
the points $\rho^{-1}(g_x)$ of a group space $G_0$ to the points
$g_x\si(x)$, $g_x\in G_x$, of a fibre $M_x$. Let us endow $G_0$
with the standard coordinate atlas $(r^\la)=(t^a,\vf^i)$ of the
group (\ref{g120a}). Then $U$ admits the trivialization
(\ref{z10'}) with respect to the bundle coordinates
$(x^A,t^a,\vf^i)$ where $x^A$, $A=1,\ldots,\di Z-m$, are
coordinates on a base $N$. The vector fields $\vt_\la$ on $U$
relative to these coordinates read
\mar{ww25}\beq
\vt_a=\dr_a, \qquad \vt_i=-(BC^{-1})^a_i(x)\dr_a +
(C^{-1})_i^k(x)\dr_k.\label{ww25}
\eeq
Accordingly, the subring $\cC$ restricted to $U$ is the pull-back
$\pi^*C^\infty(N)$ onto $U$ of a ring of smooth functions on $N$.

(III). Let us split coordinates $(x^A)$ on $N$ into some $m$
coordinates $(J_\la)$ and the rest $\di Z- 2m$ coordinates
$(z^A)$. Then we can provide the toroidal domain $U$ (\ref{z10'})
with the Poisson bivector field
\mar{kk500}\beq
w=\dr^\la\w\dr_\la \label{kk500}
\eeq
of rank $2m$. The independent complete vector fields $\dr_a$ and
$\dr_i$ are Hamiltonian vector fields of the functions $S_a=J_a$
and $S_i=J_i$ on $U$ which are in involution with respect to a
Poisson bracket $\{,\}_w$ defined by the bivector field $w$
(\ref{kk500}). By virtue of the expression (\ref{ww25}), the
Hamiltonian vector fields $\{\dr_\la\}$ generate the $\cC$-algebra
$\cA$. Therefore, $(w,\cA)$ is a commutative partially integrable
system on a Poisson manifold $(Z,w)$.
\end{proof}

\begin{remark} \label{kk501} \mar{kk501}
If fibres of a fibred manifold in item (ii) of Theorem \ref{bi0}
are assumed to be compact then this fibred manifold is a fibre
bundle and vertical vector fields on it (e.g., in condition (i) of
Theorem \ref{bi0}) are complete.
\end{remark}

A Poisson structure in Theorem \ref{bi0} is by no means unique as
follows.

\begin{theorem} \label{bi92} \mar{bi92}
Given the toroidal domain $U$ (\ref{z10'}) provided with bundle
coordinates $(x^A,r^\la)$, it is readily observed that, if a
Poisson bivector field on $U$ satisfies Definition \ref{cc1}, it
takes a form
\mar{bi20}\beq
w=w_1+w_2=w^{A\la}(x^B)\dr_A\w\dr_\la +
w^{\m\nu}(x^B,r^\la)\dr_\m\w \dr_\nu. \label{bi20}
\eeq
Conversely, given a Poisson bivector field $w$ (\ref{bi20}) of
rank $2m$ on the toroidal domain $U$ (\ref{z10'}), there exists a
toroidal domain $U'\subset U$ such that a commutative dynamical
algebra $\cA$ in Theorem \ref{bi0} is a commutative partially
integrable system on $U'$.
\end{theorem}

\begin{remark}
It is readily observed that any Poisson bivector field $w$
(\ref{bi20}) fulfils condition (b) in Definition \ref{cc1}, but
condition (a) imposes a restriction on a toroidal domain $U$. A
key point is that the characteristic foliation $\cF$ of $U$
yielded by the Poisson bivector fields $w$ (\ref{bi20}) is the
pull-back of an $m$-dimensional foliation $\cF_N$ of a base $N$,
which is defined by the first summand $w_1$ (\ref{bi20}) of $w$.
With respect to the adapted coordinates $(J_\la,z^A)$,
$\la=1,\ldots, m$, on the foliated manifold $(N,\cF_N)$, the
Poisson bivector field $w$ reads
\mar{bi42}\beq
w= w^\m_\n(J_\la,z^A)\dr^\n\w \dr_\m +
w^{\m\n}(J_\la,z^A,r^\la)\dr_\m\w \dr_\n. \label{bi42}
\eeq
Then condition (a) in Definition \ref{cc1} is satisfied if
$N'\subset N$ is a domain of a coordinate chart $(J_\la,z^A)$ of
the foliation $\cF_N$. In this case, a commutative dynamical
algebra $\cA$ on a toroidal domain $U'=\pi^{-1}(N')$ is generated
by the Hamiltonian vector fields
\mar{bi93}\beq
\vt_\la=-w\lfloor dJ_\la=w^\m_\la\dr_\m \label{bi93}
\eeq
of the $m$ independent functions $S_\la=J_\la$.
\end{remark}

\begin{proof}
The characteristic distribution of the Poisson bivector field $w$
(\ref{bi20}) is spanned by Hamiltonian vector fields
\mar{bi21}\beq
v^A=-w\lfloor dx^A=w^{A\m}\dr_\m \label{bi21}
\eeq
and vector fields
\be
w\lfloor dr^\la= w^{A\la}\dr_A + 2w^{\m\la}\dr_\m.
\ee
Since $w$ is of rank $2m$, the vector fields $\dr_\m$ can be
expressed in the vector fields $v^A$ (\ref{bi21}). Hence, the
characteristic distribution of $w$ is spanned by the Hamiltonian
vector fields $v^A$ (\ref{bi21}) and the vector fields
\mar{bi25}\beq
v^\la=w^{A\la}\dr_A. \label{bi25}
\eeq
The vector fields (\ref{bi25}) are projected onto $N$. Moreover,
one can derive from the relation $[w,w]_\mathrm{SN}=0$ that they
generate a Lie algebra and, consequently, span an involutive
distribution $\cV_N$ of rank $m$ on $N$. Let $\cF_N$ denote the
corresponding foliation of $N$. We consider the pull-back
$\cF=\pi^*\cF_N$ of this foliation onto $U$ by the trivial
fibration $\pi$. Its leaves are the inverse images $\pi^{-1}(F_N)$
of leaves $F_N$ of the foliation $\cF_N$, and so is its
characteristic distribution
\be
T\cF=(T\pi)^{-1}(\cV_N).
\ee
This distribution is spanned by the vector fields $v^\la$
(\ref{bi25}) on $U$ and the vertical vector fields on $U\to N$,
namely,  the vector fields $v^A$ (\ref{bi21}) generating a
commutative dynamical algebra $\cA$. Hence, $T\cF$ is the
characteristic distribution of a Poisson bivector field $w$.
Furthermore, since $U\to N$ is a trivial bundle, each leaf
$\pi^{-1}(F_N)$ of the pull-back foliation $\cF$ is the manifold
product of a leaf $F_N$ of $N$ and the toroidal cylinder
(\ref{g120a}). It follows that the foliated manifold $(U,\cF)$ can
be provided with an adapted coordinate atlas
\be
\{(U_\iota,J_\la,z^A,r^\la)\}, \qquad \la=1,\ldots, m, \qquad
A=1,\ldots,\di Z-2m,
\ee
such that $(J_\la,z^A)$ are adapted coordinates on the foliated
manifold $(N,\cF_N)$. Relative to these coordinates, the Poisson
bivector field (\ref{bi20}) takes the form (\ref{bi42}). Let $N'$
be the domain of this coordinate chart. Then a commutative
dynamical algebra $\cA$ on a toroidal domain $U'=\pi^{-1}(N')$ is
generated by the Hamiltonian vector fields $\vt_\la$ (\ref{bi93})
of functions $S_\la=J_\la$. {}
\end{proof}

\begin{remark} \label{kk508} \mar{kk508}
Let us note that coefficients $w^{\m\nu}$ in the expressions
(\ref{bi20}) and (\ref{bi42}) are affine in coordinates $r^\la$
because of the relation $[w,w]_\mathrm{SN}=0$ and, consequently,
they are constant on tori.
\end{remark}

Now, let $w$ and $w'$ be two different Poisson structures
(\ref{bi20}) on the toroidal domain (\ref{z10'}) which make a
commutative dynamical algebra $\cA$ into different commutative
partially integrable systems $(w,\cA)$ and $(w',\cA)$.

\begin{definition} \label{cc2} \mar{cc2}
We agree to call a triple $(w,w',\cA)$ the bi-Poisson commutative
partially integrable system if any Hamiltonian vector field
$\vt\in\cA$ with respect to $w$ possesses the same Hamiltonian
representation
\mar{bi71}\beq
\vt=-w\lfloor df=-w'\lfloor df, \qquad f\in\cC, \label{bi71}
\eeq
relative to $w'$, and \textit{vice versa}.
\end{definition}

Definition \ref{cc2} establishes \textit{sui generis} an
equivalence between the commutative partially integrable systems
$(w,\cA)$ and $(w',\cA)$.

\begin{theorem} \label{bi72} \mar{bi72}
(I) The triple $(w,w',\cA)$ is a bi-Poisson partially integrable
system in accordance with Definition \ref{cc2} iff the Poisson
bivector fields $w$ and $w'$ (\ref{bi20}) differ in the second
terms $w_2$ and $w'_2$. (II) These Poisson bivector fields admit a
recursion operator.
\end{theorem}

\begin{proof} (I). It is easily justified that, if Poisson bivector
fields $w$ (\ref{bi20}) fulfil Definition \ref{cc2}, they are
distinguished only by the second summand $w_2$. Conversely, as
follows from the proof of Theorem \ref{bi92}, the characteristic
distribution of the Poisson bivector field $w$ (\ref{bi20}) is
spanned by the vector fields (\ref{bi21}) and (\ref{bi25}). Hence,
all Poisson bivector fields $w$ (\ref{bi20}) distinguished only by
the second summand $w_2$ have the same characteristic
distribution, and they bring $\cA$ into a commutative partially
integrable system on the same toroidal domain $U'$. Then the
condition in Definition \ref{cc2} is easily justified.   (II). The
result follows from forthcoming Lemma \ref{pp1}.
\end{proof}

Given  a smooth real manifold $X$, let $w$ and $w'$ be Poisson
bivector fields of rank $2m$ on $X$, and let $w^\sh$ and $w'^\sh$
be the corresponding bundle homomorphisms (\ref{m51}). A
tangent-valued one-form $R$ on $X$ yields bundle endomorphisms
\mar{bi90}\beq
R: TX\to TX, \qquad R^*: T^*X\to T^*X. \label{bi90}
\eeq
It is called the  recursion operator if
\mar{pp0}\beq
w'^\sh=R\circ w^\sh=w^\sh\circ R^*. \label{pp0}
\eeq

\begin{lemma} \label{pp1} \mar{pp1}
A recursion operator between Poisson structures of the same rank
exists iff their characteristic distributions coincide.
\end{lemma}

\begin{proof} It follows from the equalities
(\ref{pp0}) that a recursion operator $R$ sends the characteristic
distribution of $w$ to that of $w'$, and these distributions
coincide if $w$ and $w'$ are of the same rank. Conversely, let
Poisson structures $w$ and $w'$ possess the same characteristic
distribution $T\cF\to TX$ tangent to a foliation $\cF$ of $X$. We
have the exact sequences (\ref{pp2}) -- (\ref{pp3}). The bundle
homomorphisms $w^\sh$ and $w'^\sh$ (\ref{m51}) factorize in the
unique fashion (\ref{lmp03}) through the bundle isomorphisms
$w_\cF^\sh$ and $w'^\sh_\cF$ (\ref{lmp03}). Let us consider
inverse isomorphisms
\mar{pp13}\beq
w_\cF^\fl : T\cF\to T\cF^*, \qquad w'^\fl_\cF : T\cF\to T\cF^*
\label{pp13}
\eeq
and compositions
\mar{pp10}\beq
R_\cF= w'^\sh_\cF\circ w_\cF^\fl: T\cF\to T\cF, \qquad R_\cF^*=
w_\cF^\fl \circ w'^\sh_\cF: T\cF^*\to T\cF^*. \label{pp10}
\eeq
There is an obvious relation
\be
w'^\sh_\cF=R_\cF\circ w^\sh_\cF=  w^\sh_\cF\circ R^*_\cF.
\ee
In order to obtain a recursion operator (\ref{pp0}), it suffices
to extend the morphisms $R_\cF$ and $R_\cF^*$ (\ref{pp10}) onto
$TX$ and $T^*X$, respectively. For this purpose, let us consider a
splitting
\be
\zeta: TX\to T\cF, \qquad TX=T\cF\oplus
(\id-i_\cF\circ\zeta)TX=T\cF\oplus E,
\ee
of the exact sequence (\ref{pp2}) and the dual splitting
\be
 \zeta^* :T\cF^*\to T^*X, \quad
T^*X=\zeta^*(T\cF^*)\oplus (\id-\zeta^*\circ i^*_\cF)T^*X=
\zeta^*(T\cF^*)\oplus E',
\ee
of the exact sequence (\ref{pp3}). Then the desired extensions are
\be
R=R_\cF\times \id E, \qquad R^*=(\zeta^*\circ R^*_\cF)\times \id
E'.
\ee
This recursion operator is invertible, i.e., the morphisms
(\ref{bi90}) are bundle isomorphisms.
\end{proof}

For instance,  the Poisson bivector field $w$ (\ref{bi20}) and the
Poisson bivector field
\be
w_0=w^{A\la}\dr_A\w\dr_\la
\ee
admit a recursion operator $w^\sh_0=R\circ w^\sh$ whose entries
are given by the equalities
\be
R^A_B=\dl^A_B, \qquad R^\m_\n=\dl^\m_\n, \qquad R^A_\la=0, \qquad
w^{\m\la}=R^\la_Bw^{B\m}.
\ee

Given a commutative partially integrable system $(w,\cA)$ in
Theorem \ref{bi92}, the bivector field $w$ (\ref{bi42}) can be
brought into the canonical form (\ref{kk500}) with respect to
partial action-angle coordinates in forthcoming Theorem
\ref{bi100}. This theorem extends the Liouville--Arnold theorem to
the case of a Poisson structure and a non-compact invariant
submanifold \cite{jmp03,book10,book15}.

\begin{theorem} \label{bi100} \mar{bi100}
Given a commutative partially integrable system $(w,\cA)$ on a
Poisson manifold $(U,w)$, there exists a toroidal domain
$U'\subset U$ equipped with partial  action-angle coordinates
$(I_a,I_i,z^A, \tau^a,\f^i)$ such that, restricted to $U'$, a
Poisson bivector field takes the canonical form
\mar{bi101}\beq
w=\dr^a\w \dr_a + \dr^i\w \dr_i, \label{bi101}
\eeq
while a commutative dynamical algebra $\cA$ is generated by
Hamiltonian vector fields of the action coordinate functions
$S_a=I_a$, $S_i=I_i$.
\end{theorem}

\begin{proof}
First, let us employ Theorem \ref{bi92} and restrict $U$ to a
toroidal domain, say  $U$ again, equipped with coordinates
$(J_\la,z^A,r^\la)$ such that a Poisson bivector field $w$ takes
the form (\ref{bi42}) and a commutative dynamical algebra $\cA$ is
generated by the Hamiltonian vector fields $\vt_\la$ (\ref{bi93})
of $m$ independent functions $S_\la=J_\la$ in involution. Let us
choose these vector fields as new generators of a group $G$ and
return to Theorem \ref{bi0}. In accordance with this theorem,
there exists a toroidal domain $U'\subset U$ provided with another
trivialization $U'\to N'\subset N$ in the toroidal cylinders
(\ref{g120a}) and endowed with bundle coordinates
$(J_\la,z^A,r^\la)$ such that the vector fields $\vt_\la$
(\ref{bi93}) take the form (\ref{ww25}). For the sake of
simplicity, let $U'$, $N'$ and $y^\la$ be denoted $U$, $N$ and
$r^\la=(t^a,\vf^i)$ again. Herewith, a Poisson bivector field $w$
is given by the expression (\ref{bi42}) with new coefficients. Let
$w^\sh: T^*U\to TU$ be the corresponding bundle homomorphism. It
factorizes in a unique fashion (\ref{lmp03}):
\be
w^\sh: T^*U\ar^{i^*_\cF} T\cF^*\ar^{w^\sh_\cF} T\cF\ar^{i_\cF} TU
\ee
through the bundle isomorphism
\be
w_\cF^\sh: T\cF^*\to T\cF,  \qquad w^\sh_\cF:\al\to -w(x)\lfloor
\al.
\ee
Then the inverse isomorphisms $w_\cF^\fl : T\cF\to T\cF^*$
provides a foliated manifold $(U,\cF)$ with the leafwise
symplectic form
\mar{bi102,'}\ben
&& \Om_\cF=\Om^{\m\n}(J_\la,z^A,t^a) \wt dJ_\m\w \wt dJ_\n +
\Om_\m^\n(J_\la,z^A) \wt dJ_\n\w \wt dr^\m, \label{bi102}\\
&& \Om_\m^\al w^\m_\bt=\dl^\al_\bt, \qquad
\Om^{\al\bt}=-\Om^\al_\m\Om^\bt_\n w^{\m\n}. \label{bi102'}
\een
Let us show that it is $\wt d$-exact. Let $F$ be a leaf of the
foliation $\cF$ of $U$. There is a homomorphism of the de Rham
cohomology $H^*_{\rm DR}(U)$ of $U$ to the de Rham cohomology
$H^*_{\rm DR}(F)$ of $F$, and it factorizes through the leafwise
cohomology $H^*_\cF(U)$. Since $N$ is a domain of an adapted
coordinate chart of the foliation $\cF_N$, the foliation $\cF_N$
of $N$ is a trivial fibre bundle
\be
N=V\times W\to W.
\ee
Since $\cF$ is the pull-back onto $U$ of the foliation $\cF_N$ of
$N$, it also is a trivial fibre bundle
\mar{bi103}\beq
U=V\times W\times (\mathbb R^{k-m}\times T^m) \to W \label{bi103}
\eeq
over a domain $W\subset \mathbb R^{\di Z-2m}$. It follows that
\be
H^*_{\rm DR}(U)=H^*_{\rm DR}(T^r)=H^*_\cF(U).
\ee
Then the closed leafwise two-form $\Om_\cF$ (\ref{bi102}) is exact
due to the absence of the term $\Om_{\m\n}dr^\m\w dr^\nu$.
Moreover, $\Om_\cF=\wt d\Xi$ where $\Xi$ reads
\be
\Xi=\Xi^\al(J_\la,z^A,r^\la)\wt dJ_\al + \Xi_i(J_\la,z^A)\wt
d\vf^i
\ee
up to a $\wt d$-exact leafwise form. The Hamiltonian vector fields
$\vt_\la=\vt_\la^\m\dr_\m$ (\ref{ww25}) obey the relation
\mar{ww22'}\beq
\vt_\la\rfloor\Om_\cF=-\wt dJ_\la, \qquad \Om^\al_\bt
\vt^\bt_\la=\dl^\al_\la, \label{ww22'}
\eeq
which falls into the following conditions
\mar{bi110,1}\ben
&& \Om^\la_i=\dr^\la\Xi_i-\dr_i\Xi^\la, \label{bi110} \\
&& \Om^\la_a=-\dr_a\Xi^\la=\dl^\la_a. \label{bi111}
\een
The first of the relations (\ref{bi102'})  shows that
$\Om^\al_\bt$ is a non-degenerate matrix  independent of
coordinates $r^\la$. Then the condition (\ref{bi110}) implies that
$\dr_i\Xi^\la$ are  independent of $\vf^i$, and so are $\Xi^\la$
since $\vf^i$ are cyclic coordinates. Hence,
\mar{bi112,3}\ben
&&\Om^\la_i=\dr^\la\Xi_i, \label{bi112}\\
&& \dr_i\rfloor\Om_\cF=-\wt d\Xi_i. \label{bi113}
\een
Let us introduce new coordinates $I_a=J_a$, $I_i=\Xi_i(J_\la)$. By
virtue of the equalities (\ref{bi111}) and (\ref{bi112}), the
Jacobian of this coordinate transformation is regular. The
relation (\ref{bi113}) shows that $\dr_i$ are Hamiltonian vector
fields of the functions $S_i=I_i$. Consequently, we can choose
vector fields $\dr_\la$ as generators of a commutative dynamical
algebra $\cA$. One obtains from the equality (\ref{bi111}) that
\be
\Xi^a=-t^a+E^a(J_\la,z^A)
\ee
and $\Xi^i$ are independent of $t^a$. Then the leafwise Liouville
form $\Xi$ reads
\be
\Xi=(-t^a+E^a(I_\la,z^A))\wt dI_a + E^i(I_\la, z^A)\wt dI_i + I_i
\wt d\vf^i.
\ee
The coordinate shifts
\be
\tau^a=-t^a+E^a(I_\la,z^A), \qquad \f^i=\vf^i-E^i(I_\la,z^A)
\ee
bring the leafwise form $\Om_\cF$ (\ref{bi102}) into the canonical
form
\be
\Om_\cF= \wt dI_a\w \wt d \tau^a + \wt dI_i\w \wt d\f^i
\ee
which ensures the canonical form (\ref{bi101}) of a Poisson
bivector field $w$.
\end{proof}

\section{Partially superintegrable systems on Poisson manifolds}

Studying partially superintegrable systems, we bear in mind that
in Example \ref{x6}, but restrict our consideration to Lie algebra
superintegrable systems whose generating functions constitute a
real Lie algebra (Section 3).

Given a smooth manifold $Z$, let $\{\vt_i\}$ be $k$ independent
vector fields on $Z$ (i.e., $\op\w^k \vt_i\neq 0$) which generate
a real Lie algebra $\cG$ with commutation relations
\mar{x9}\beq
[\vt_i,\vt_j]= c_{ij}^h \vt_h. \label{x9}
\eeq
We denote by $\cC\subset C^\infty(Z)$ a $\mathbb R$-subring of
smooth real functions $f$ on $Z$ whose derivations $\vt_i\rfloor
df$ vanish for all $\vt_i$. Let $\cA$ be a $k$-dimensional Lie
$\cC$-algebra generated by vector fields $\{\vt_i\}$.

\begin{definition} \label{x10} \mar{x10}
We agree to call $\cA$ the dynamical algebra.
\end{definition}

In particular, this definition reproduces Definition \ref{kk1} of
a commutative dynamical algebra if vector fields $\{\vt_i\}$
mutually commute.

\begin{definition} \label{x11} \mar{x11} Let $(Z,w)$ be a
Poisson manifold and $\cA$ a $k$-dimensional dynamical algebra on
$Z$. A triple $(Z,\cA,w)$ is said to be a partially
superintegrable system if the following hold.

(a) Generators $\vt_i$ of $\cA$ are Hamiltonian vector fields of
some independent functions $F_i$ on $Z$. In view of the relations
(\ref{m81}), these functions obey the commutation relations
\mar{x12}\beq
\{F_i,F_j\}_w= c_{ij}^h F_h, \qquad \{F_i, f\}_w=0, \qquad
f\in\cC. \label{x12}
\eeq

(b) All elements of $\cC\subset C^\infty(Z)$ are mutually in
involution, i.e., their Poisson brackets $\{f,f'\}_w$ equal zero.
\end{definition}

For instance, let $F=(F_i)$ be the Lie algebra superintegrable
system (\ref{zz60}) on a symplectic manifold $(Z,\Om)$.
Hamiltonian vector fields $\vt_i$ of its generating functions
$F_i$ obey the commutation relations (\ref{x9}) and yield a
dynamical algebra. Then it is a partially superintegrable system
in accordance with Definition \ref{x11} where a ring $\cC$
consists of the pull-back of Casimir functions of the Lie--Poisson
structure on a Lie coalgebra $\cG^*$.

A partially superintegrable systems on a product of manifolds in
Examples \ref{x6} and commutative partially integrable systems in
Example \ref{x7} also are well in accordance with Definition
\ref{x11}.

Certainly, a Poisson structure $w$ in Definition \ref{x11} is not
unique (see Theorem \ref{bi92} for a case of commutative partially
integrable systems). Following Definition \ref{kk2}, we agree to
call it compatible. Generalizing Theorem \ref{bi72}, one can show
that two Poisson structures are compatible only if they admit the
recursion operator (\ref{bi90}). In accordance with Lemma
\ref{pp1}, their symplectic foliations coincide.

\begin{lemma} \label{x50} \mar{x50}
The rank of a compatible Poisson structure of a partially
superintegrable system equals $k+m$ where $m$ is a corank of the
Lie algebra $\cG$ (\ref{x9}).
\end{lemma}

\begin{proof}
Let $\cW$ be a symplectic foliation of $Z$ and $W$ its leaf. The
generating vector field $\vt_i$ obviously are tangent to $W$.
Restricted to $W$, they generate a Lie algebra superintegrable
system so that $k+m$ is a dimension of a leaf $W$.
\end{proof}

Let $(Z,\cA,w)$ be a partially superintegrable system, and let its
generating vector fields be complete. In accordance with the
above-mentioned theorem \cite{onish,palais}, they define a
Hamiltonian action of a simply connected Lie group $G$ whose Lie
algebra is isomorphic to $\cG$ on $Z$. Since vector fields $\vt_i$
are independent, the action of $G$ on $Z$ is locally free, i.e.,
isotropy groups of points of $U$ are discrete subgroups of $G$.
Orbits of $G$ coincide with $k$-dimensional maximal integral
manifolds of a regular distribution $\cV$ on $Z$ spanned by vector
fields $\vt_i$ \cite{susm}. They constitute a foliation $\cF$ of
$Z$. It is subordinate a symplectic foliation $\cW$ whose leaves
are foliated by leaves of $\cF$.

Let both a foliation $\cF$ and a foliation $\cW$ be fibred
manifolds $\pi_\cF$ and $\pi_\cW$ with mutually diffeomorphic
fibres, respectively. Thus, we have a composite fibred manifold
\mar{x40}\beq
\pi_\cW: Z\ar \pi_\cF(Z)\ar \pi_\cW(Z). \label{x40}
\eeq
Then one can show the following.

\begin{theorem} \label{x51} \mar{x51}
Let $V$ be a leaf of $\cF$. Then there exists an open saturated
neighborhood $U_V$ of $V$ such that the composite fibred manifold
$\pi_\cW$ (\ref{x40}) becomes a composite bundle
\beq
\pi_\cW: U_V\ar \pi_\cF(U_V)\ar \pi_\cW(U_V) \label{x41}
\eeq
which is endowed with bundle (generalized action-angle)
coordinates
\be
(I_\la,q^A,y^\la,w^a)\to (I_\la,q^A,x^a)\to (I_\la,x^a), \quad
\la=1,\ldots,m, \quad A=1,\ldots, k-m,
\ee
where: (i) the angle coordinates $y^i$ are coordinates on $\mathbb
R^{m-r}\times T^r$; (ii) the $(I_\la, x^a)$ are coordinates on
$\pi_\cF(U_V)$; the $(x^a)$ are coordinates on $\pi_\cW(U_V)$.
With respect to these coordinates, the Poisson bivector field
takes a form
\be
w= \dr^\la\w \dr_\la + w^{A\la}(q^B,x^a)\dr_A\w\dr_B.
\ee
\end{theorem}

\begin{proof} The proof of Theorem \ref{x51} is reduced to Theorems
\ref{bi92} and \ref{bi100} for commutative partially integrable
systems because the pull-back $\pi_\cF^*C$ of functions on
$\pi_\cF(U_N)$ constitute a commutative partially integrable
system on a Poisson manifold $(Z,w)$ whose integral manifolds,
diffeomorphic to toroidal cylinders, are invariant submanifolds of
fibres of $\pi_\cF$.
\end{proof}

\section{Appendix}

For the convenience of the reader this Section summarize the
relevant mathematical material on symplectic manifolds, Poisson
manifolds and symplectic foliations
\cite{abr,book05,book10,libe,vais}.

\subsection{Distributions and foliations}

A subbundle $\bT$ of the tangent bundle $TZ$ of a manifold $Z$ is
called a regular  distribution (or, simply, a distribution). A
vector field $u$ on $Z$ is said to be  subordinate to a
distribution $\bT$ if it lives in $\bT$. A distribution $\bT$ is
called  involutive if the Lie bracket of $\bT$-subordinate vector
fields also is subordinate to $\bT$.

A subbundle of the cotangent bundle $T^*Z$ of $Z$ is called a
codistribution $\bT^*$ on a manifold $Z$. For instance, the
annihilator  $\rA\bT$ of a distribution $\bT$ is a codistribution
whose fibre over $z\in Z$ consists of covectors $w\in T^*_z$ such
that $v\rfloor w=0$ for all $v\in \bT_z$.

The following local coordinates can be associated to an involutive
distribution  \cite{war}.

\begin{theorem}\label{c11.0} \mar{c11.0} Let $\bT$ be an involutive
$r$-dimensional distribution on a manifold $Z$, $\di Z=k$. Every
point $z\in Z$ has an open neighborhood $U$ which is a domain of
an adapted coordinate chart $(z^1,\dots,z^k)$ such that,
restricted to $U$, the distribution $\bT$ and its annihilator
$\rA\bT$ are spanned by the local vector fields $\dr/\dr z^1,
\cdots,\dr/\dr z^r$ and the local one-forms $dz^{r+1},\dots,
dz^k$, respectively.
\end{theorem}

A connected submanifold $N$ of a manifold $Z$ is called an
integral manifold of a distribution $\bT$ on $Z$ if $TN\subset
\bT$. Unless otherwise stated, by an integral manifold is meant an
integral manifold of dimension of $\bT$. An integral manifold is
called  maximal if no different integral manifold contains it. The
following is the classical theorem of Frobenius \cite{kob,war}.

\begin{theorem}\label{to.1}  \mar{to.1} Let $\bT$ be an
involutive distribution on a manifold $Z$. For any $z\in Z$, there
exists a unique maximal integral manifold of $\bT$ through $z$,
and any integral manifold through $z$ is its open subset.
\end{theorem}

Maximal integral manifolds of an involutive distribution on a
manifold $Z$ are assembled into a regular foliation $\cF$ of $Z$.
A regular $r$-dimensional  foliation (or, simply, a foliation)
$\cF$ of a $k$-dimensional manifold $Z$ is defined as a partition
of $Z$ into connected $r$-dimensional submanifolds (the leaves of
a foliation) $F_\iota$, $\iota\in I$, which possesses the
following properties \cite{rei,tam}.

A manifold $Z$ admits an adapted coordinate atlas
\mar{spr850}\beq
\{(U_\xi;z^\la, z^i)\},\quad \la=1,\ldots,k-r, \qquad
i=1,\ldots,r, \label{spr850}
\eeq
such that transition functions of coordinates $z^\la$ are
independent of the remaining coordinates $z^i$. For each leaf $F$
of a foliation $\cF$, the connected components of $F\cap U_\xi$
are given by the equations $z^\la=$const. These connected
components and coordinates $(z^i)$ on them make up a coordinate
atlas of a leaf $F$. It follows that tangent spaces to leaves of a
foliation $\cF$ constitute an involutive distribution $T\cF$ on
$Z$, called the  tangent bundle to the foliation $\cF$. The factor
bundle $V\cF=TZ/T\cF$, called the  normal bundle to $\cF$, has
transition functions independent of coordinates $z^i$. Let
$T\cF^*\to Z$ denote the dual of $T\cF\to Z$. There are the exact
sequences
\mar{pp2,3}\ben
&& 0\to T\cF \ar^{i_\cF} TX \ar V\cF\to 0, \label{pp2} \\
&& 0\to {\rm Ann}\,T\cF\ar T^*X\ar^{i^*_\cF} T\cF^* \to 0
\label{pp3}
\een
of vector bundles over $Z$.

A pair $(Z,\cF)$, where $\cF$ is a foliation of $Z$, is called a
 foliated manifold. It should be emphasized that leaves of a
foliation need not be closed or imbedded submanifolds. Every leaf
has an open  saturated neighborhood $U$, i.e., if $z\in U$, then a
leaf through $z$ also belongs to $U$.

Any submersion $\zeta:Z\to M$ yields a foliation
\be
\cF=\{F_p=\zeta^{-1}(p)\}_{p\in \zeta(Z)}
\ee
of $Z$ indexed by elements of $\zeta(Z)$, which is an open
submanifold of $M$, i.e., $Z\to \zeta(Z)$ is a fibred manifold.
Leaves of this foliation are closed imbedded submanifolds. Such a
foliation is called  simple. Any (regular) foliation is locally
simple.

\begin{example} \label{10b6} \mar{10b6}
Every smooth real function $f$ on a manifold $Z$ with nowhere
vanishing differential $df$ is a submersion $Z\to \mathbb R$. It
defines a one-codimensional foliation whose leaves are given by
equations
\be
f(z)=c, \qquad c\in f(Z)\subset\mathbb R.
\ee
This is a foliation of level surfaces of a function $f$, called
the generating function. Every one-codimensional foliation is
locally a foliation of level surfaces of some function on $Z$. The
level surfaces of an arbitrary smooth real function $f$ on a
manifold $Z$ define a  singular foliation $\cF$ on $Z$
\cite{kamb}. Its leaves are not submanifolds in general.
Nevertheless if $df(z)\neq 0$, the restriction of $\cF$ to some
open neighborhood $U$ of $z$ is a foliation with the generating
function $f|_U$.
\end{example}

Let $\cF$ be a (regular) foliation of a $k$-dimensional manifold
$Z$ provided with the adapted coordinate atlas (\ref{spr850}). The
real Lie algebra $\cT_1(\cF)$ of  global sections of the tangent
bundle $T\cF\to Z$ to $\cF$ is a $C^\infty(Z)$-submodule of the
derivation module of the $\mathbb R$-ring $C^\infty(Z)$ of smooth
real functions on $Z$. Its kernel $S_\cF(Z)\subset C^\infty(Z)$
consists of functions constant on leaves of $\cF$. Therefore,
$\cT_1(\cF)$ is the Lie $S_\cF(Z)$-algebra of derivations of
$C^\infty(Z)$, regarded as a $S_\cF(Z)$-ring. Then one can
introduce the leafwise differential calculus \cite{jmp02,book05}
as the Chevalley--Eilenberg differential calculus over the
$S_\cF(Z)$-ring $C^\infty(Z)$. It is defined as a subcomplex
\mar{spr892}\beq
0\to S_\cF(Z)\ar C^\infty(Z)\ar^{\wt d} \gF^1(Z) \cdots \ar^{\wt
d} \gF^{\di\cF}(Z) \to 0 \label{spr892}
\eeq
of the Chevalley--Eilenberg complex of the Lie $S_\cF(Z)$-algebra
$\cT_1(\cF)$ with coefficients in $C^\infty(Z)$ which consists of
$C^\infty(Z)$-multilinear skew-symmetric maps
\be
\op\times^r \cT_1(\cF) \to C^\infty(Z), \qquad r=1,\ldots,\di\cF.
\ee
These maps are global sections of exterior products $\op\w^r
T\cF^*$ of the dual $T\cF^*\to Z$ of $T\cF\to Z$. They are called
the leafwise forms on a foliated manifold $(Z,\cF)$, and are given
by the coordinate expression
\be
\f=\frac1{r!}\f_{i_1\ldots i_r}\wt dz^{i_1}\w\cdots\w \wt
dz^{i_r},
\ee
where $\{\wt dz^i\}$ are the duals of the holonomic fibre bases
$\{\dr_i\}$ for $T\cF$. Then one can think of the Chevalley --
Eilenberg coboundary operator
\be
\wt d\f= \wt dz^k\w \dr_k\f=\frac{1}{r!} \dr_k\f_{i_1\ldots
i_r}\wt dz^k\w\wt dz^{i_1}\w\cdots\w\wt dz^{i_r}
\ee
as being the  leafwise exterior differential. Accordingly, the
complex (\ref{spr892}) is called the  leafwise de Rham complex (or
the tangential de Rham complex).

Let us consider the exact sequence (\ref{pp3}) of vector bundles
over $Z$. Since it admits a splitting, the epimorphism $i^*_\cF$
yields that of the algebra $\cO^*(Z)$ of exterior forms on $Z$ to
the algebra $\gF^*(Z)$ of leafwise forms. It obeys the condition
$i^*_\cF\circ d=\wt d\circ i^*_\cF$, and provides the cochain
morphism
\mar{lmp04}\ben
&& i^*_\cF: (\mathbb R,\cO^*(Z),d)\to (S_\cF(Z),\cF^*(Z),\wt d),
\label{lmp04}\\
&& dz^\la\to 0, \quad dz^i\to\wt dz^i, \nonumber
\een
of the de Rham complex of $Z$ to the leafwise de Rham complex
(\ref{spr892}).

Given a leaf $i_F:F\to Z$ of $\cF$, we have the pull-back
homomorphism
\mar{lmp30}\beq
(\mathbb R,\cO^*(Z),d) \to (\mathbb R,\cO^*(F),d) \label{lmp30}
\eeq
of the de Rham complex of $Z$ to that of $F$.

\begin{proposition} \label{lmp32} \mar{lmp32} The homomorphism (\ref{lmp30})
factorize through the homomorphism \cite{book05}.
\end{proposition}

\subsection{Differential geometry of Lie groups}

Let $G$ be a real Lie group of $\di G >0$, and let $L_g:G\to gG$
and $R_g:G\to Gg$ denote the action of $G$ on itself by left and
right multiplications, respectively. Clearly, $L_g$ and $R_{g'}$
for all $g,g'\in G$ mutually commute, and so do the tangent maps
$TL_g$ and $TR_{g'}$.

A vector field $\xi_l$ (resp. $\xi_r$) on a group $G$ is said to
be left-invariant (resp.  right-invariant) if $\xi_l\circ L_g
=TL_g\circ \xi_l$ (resp. $\xi_r\circ R_g=TR_g\circ \xi_r$).
Left-invariant (resp. right-invariant) vector fields make up the
 left Lie algebra $\cG_l$ (resp. the  right Lie algebra
$\cG_r$) of  $G$.

There is one-to-one correspondence between the left-invariant
vector field $\xi_l$ (resp. right-invariant vector fields $\xi_r$)
on $G$ and the vectors $\xi_l(e)=TL_{g^{-1}}\xi_l(g)$ (resp.
$\xi_r(e)=TR_{g^{-1}}\xi_l(g)$) of the tangent space $T_eG$ to $G$
at the unit element $e$ of $G$. This correspondence provides
$T_eG$ with the left and the right Lie algebra structures.
Accordingly, the left action $L_g$ of a Lie group $G$ on itself
defines its  adjoint representation
\mar{1210}\beq
\xi_r \to {\rm Ad}\,g(\xi_r)=TL_g\circ\xi_r\circ L_{g^{-1}}
\label{1210}
\eeq
in the right Lie algebra $\cG_r$.

Let $\{\e_m\}$ (resp. $\{\ve_m\}$) denote the basis for the left
(resp. right) Lie algebra, and let $c^k_{mn}$ be the  right
structure constants
\be
[\ve_m,\ve_n]=c^k_{mn}\ve_k.
\ee
There is the morphism
\be
\rho: \cG_l\ni \e_m \to\ve_m= -\e_m\in \cG_r
\ee
between left and right Lie algebras such that
\be
[\e_m,\e_n]=-c^k_{mn}\e_k.
\ee

The tangent bundle $\pi_G:TG\to G$ of a Lie group $G$ is trivial.
There are the following two canonical isomorphisms
\be
&& \vr_l: TG\ni q\to (g=\pi_G(q), TL^{-1}_g(q))\in G\times \cG_l,\\
&& \vr_r: TG\ni q\to (g=\pi_G(q), TR^{-1}_g(q))\in G\times \cG_r.
\ee
Therefore, any action
\be
G\times Z\ni (g,z)\to gz\in Z
\ee
of a Lie group $G$ on a manifold $Z$ on the left yields the
homomorphism
\mar{spr941}\beq
\cG_r\ni\ve\to \xi_\ve\in \cT_1(Z) \label{spr941}
\eeq
of the right Lie algebra $\cG_r$ of $G$ into the Lie algebra of
vector fields on $Z$ such that
\mar{z212}\beq
\xi_{{\rm Ad}\,g(\ve)}=Tg\circ \xi_\ve\circ g^{-1}. \label{z212}
\eeq
Vector fields $\xi_\ve$ are said to be the  infinitesimal
generators of a representation  of the Lie group $G$ in $Z$.

In particular, the adjoint representation (\ref{1210}) of a Lie
group in its right Lie algebra yields the adjoint representation
\be
\ve': \ve\to {\rm ad}\,\ve' (\ve)=[\ve',\ve], \qquad {\rm
ad}\,\ve_m(\ve_n)=c^k_{mn}\ve_k,
\ee
of the right Lie algebra $\cG_r$ in itself.

The dual $\cG^*=T^*_eG$ of the tangent space $T_eG$ is called the
 Lie coalgebra). It is provided with the basis $\{\ve^m\}$
which is the dual of the basis $\{\ve_m\}$ for $T_eG$. The group
$G$ and the right Lie algebra $\cG_r$ act on $\cG^*$ by the
coadjoint representation
\mar{z211}\ben
&&\lng{\rm Ad}^*g(\ve^*),\ve\rng =\lng\ve^*,{\rm
Ad}\,g^{-1}(\ve)\rng, \quad
\ve^*\in \cG^*, \quad \ve\in \cG_r, \label{z211}\\
&&\lng{\rm ad}^*\ve'(\ve^*),\ve\rng=-\lng\ve^*,[\ve',\ve]\rng, \qquad
\ve'\in \cG_r,\nonumber\\
&&  {\rm ad}^*\ve_m(\ve^n)=-c^n_{mk}\ve^k.\nonumber
\een

The Lie coalgebra $\cG^*$ of a Lie group $G$ is provided with the
canonical Poisson structure, called the  Lie--Poisson structure
\cite{abr,libe}. It is given by the bracket
\mar{gm510}\beq
\{f,g\}_{\rm LP}=\langle\ve^*,[df(\ve^*),dg(\ve^*)]\rangle, \qquad
f,g\in C^\infty(\cG^*), \label{gm510}
\eeq
where $df(\ve^*),dg(\ve^*)\in \cG_r$ are seen as linear mappings
from $T_{\ve^*}\cG^*=\cG^*$ to $\mathbb R$. Given coordinates
$z_k$ on $\cG^*$ with respect to the basis $\{\ve^k\}$, the Lie --
Poisson bracket (\ref{gm510}) and the corresponding Poisson
bivector field $w$ read
\be
\{f,g\}_{\rm LP}=c^k_{mn}z_k\dr^mf\dr^ng, \qquad
w_{mn}=c^k_{mn}z_k.
\ee
One can show that symplectic leaves of the Lie--Poisson structure
on the coalgebra $\cG^*$ of a connected Lie group $G$ are orbits
of the coadjoint representation (\ref{z211}) of $G$ on $\cG^*$
\cite{wein}.

\subsection{Symplectic structure}

 Let $Z$ be a smooth
manifold. Any exterior two-form $\Om$ on $Z$ yields a linear
bundle morphism
\mar{m52}\beq
\Om^\fl: TZ\op\to_Z T^*Z, \qquad \Om^\fl:v \to - v\rfloor\Om(z),
\quad v\in T_zZ, \quad z\in Z. \label{m52}
\eeq
One says that a two-form $\Om$ is of rank  $r$ if the morphism
(\ref{m52}) has a rank $r$. A kernel
 $\Ker \Om$ of $\Om$
is defined as the kernel of the morphism (\ref{m52}). In
particular, $\Ker \Om$ contains the canonical zero section $\wh 0$
of $TZ\to Z$. If $\Ker \Om=\wh 0$, a two-form $\Om$ is said to be
non-degenerate. A closed non-degenerate two-form $\Om$ is called
symplectic. Accordingly, a manifold equipped with a symplectic
form is a symplectic manifold. A symplectic manifold $(Z,\Om)$
always is even dimensional and orientable.

A manifold morphism $\zeta$ of a symplectic manifold $(Z,\Om)$ to
a symplectic manifold $(Z',\Om')$ is called symplectic if
$\Om=\zeta^*\Om'$. Any symplectic morphism is an immersion. A
symplectic isomorphism is called the symplectomorphism.

A vector field $u$ on a symplectic manifold $(Z,\Om)$ is an
infinitesimal generator of a local one-parameter group of local
symplectomorphism iff the Lie derivative $\bL_u\Om$ vanishes. It
is called  the canonical vector field. A canonical vector field
$u$ on a symplectic manifold $(Z,\Om)$ is said to be Hamiltonian
if a closed one-form $u\rfloor\Om$ is exact. Any smooth function
$f\in C^\infty(Z)$ on $Z$ defines a unique Hamiltonian vector
field $\vt_f$ such that
\mar{z100}\beq
\vt_f\rfloor\Om=-df, \qquad \vt_f=\Om^\sh(df), \label{z100}
\eeq
where $\Om^\sh$ is the inverse isomorphism to $\Om^\fl$
(\ref{m52}).

\begin{example}\label{e11.0} \mar{e11.0} Given an $m$-dimensional manifold $M$ coordinated by
$(q^i)$, let
\be
\pi_{*M}:T^*M\to M
\ee
be its cotangent bundle equipped with the holonomic coordinates
$(q^i, p_i=\dot q_i)$. It is endowed with the canonical Liouville
form
\be
\Xi=p_idq^i
\ee
and the canonical symplectic form
\mar{m83}\beq
\Om_T= d\Xi=dp_i\w dq^i. \label{m83}
\eeq
Their coordinate expressions are maintained under holonomic
coordinate transformations. The Hamiltonian vector field $\vt_f$
(\ref{z100}) with respect to the canonical symplectic form
(\ref{m83}) reads
\be
\vt_f=\dr^if\dr_i -\dr_i f\dr^i.
\ee
\end{example}

The canonical symplectic form (\ref{m83}) plays a prominent role
in symplectic geometry in view of the classical Darboux theorem.

\begin{theorem} \label{spr825} \mar{spr825}
Each point of a symplectic manifold $(Z,\Om)$ has an open
neighborhood equipped with coordinates $(q^i,p_i)$, called
canonical or Darboux coordinates, such that $\Om$ takes the
coordinate form (\ref{m83}).
\end{theorem}

Let $i_N:N\to Z$ be a submanifold of a $2m$-dimensional symplectic
manifold $(Z,\Om)$. A subset
\be
\rO_\Om TN=\op\bigcup_{z\in N} \{v\in T_zZ: \, v\rfloor u\rfloor
\Om=0, \,\,  u\in T_zN\}
\ee
of $TZ|_N$ is called orthogonal to $TN$ relative to a symplectic
form $\Om$. One considers the following special types of
submanifolds of a symplectic manifold such that the pull-back
$\Om_N=i_N^*\Om$ of a symplectic form $\Om$ onto a submanifold $N$
is of constant rank. A submanifold $N$ of $Z$ is said to be:

$\bullet$  coisotropic if $\rO_\Om TN\subseteq TN$, $\di N\geq m$;

$\bullet$ symplectic if $\Om_N$ is a symplectic form on $N$;

$\bullet$ isotropic  if $TN\subseteq \rO_\Om TN$, $\di N\leq m$.

\subsection{Poisson structure}

A Poisson bracket on a ring $C^\infty(Z)$ of smooth real functions
on a manifold $Z$ (or a Poisson structure on $Z$) is defined as an
$\mathbb R$-bilinear map
\be
C^\infty(Z)\times C^\infty(Z)\ni (f,g)\to \{f,g\}\in C^\infty(Z)
\ee
which satisfies the following conditions:

$\bullet$ $\{g,f\}=-\{f,g\}$;

$\bullet$ $\{f,\{g,h\}\} + \{g,\{h,f\}\} +\{h,\{f,g\}\}=0$;

$\bullet$ $\{h,fg\}=\{h,f\}g +f\{h,g\}$.

A  Poisson bracket makes $C^\infty(Z)$ into a real Lie algebra,
called the Poisson algebra.  A Poisson structure is characterized
by a particular bivector field as follows.

\begin{theorem}\label{p11.1} \mar{p11.1} Every Poisson bracket on a
manifold $Z$ is uniquely defined as
\mar{b450}\beq
\{f,f'\}=w(df,df')=w^{\m\nu}\dr_\m f\dr_\nu f' \label{b450}
\eeq
by a bivector field $w$ whose Schouten--Nijenhuis bracket
$[w,w]_{\rm SN}$ vanishes. It is called a Poisson bivector field.
\end{theorem}

A manifold $Z$ endowed with a Poisson structure is called a
Poisson manifold.

\begin{example} \label{w206} \mar{w206}
Any manifold admits a zero Poisson structure characterized by a
zero Poisson bivector field $w=0$.
\end{example}

Any bivector field $w$ on a manifold $Z$ yields a linear bundle
morphism
\mar{m51}\beq
w^\sh: T^*Z\op\to_Z TZ, \qquad w^\sh:\al\to -w(z)\lfloor \al,
\qquad \al\in T^*_zZ. \label{m51}
\eeq
One says that $w$ is of rank  $r$ if the morphism (\ref{m51}) is
of this rank. If a Poisson bivector field is of constant rank, the
Poisson structure is called regular. Throughout this work, only
regular Poisson structures are considered. A Poisson structure
determined by a Poisson bivector field $w$ is said to be
non-degenerate if $w$ is of maximal rank.

There is one-to-one correspondence $\Om_w \leftrightarrow w_\Om$
between the symplectic forms and the non-degenerate Poisson
bivector fields  which is given by the equalities
\be
&&w_\Om(\f,\si)=\Om_w(w_\Om^\sh(\f),w_\Om^\sh(\si)), \qquad
\f,\si\in\cO^1(Z),\\
&&\Om_w(\vt,\nu)=w_\Om(\Om_w^\fl(\vt),\Om_w^\fl(\nu)), \qquad
\vt,\nu\in \cT(Z),
\ee
where the morphisms $w_\Om^\sh$ (\ref{m51}) and $\Om_w^\fl$
(\ref{m52}) are mutually inverse.

However, this correspondence is not preserved under manifold
morphisms in general. Namely, let $(Z_1,w_1)$ and $(Z_2,w_2)$ be
Poisson manifolds. A manifold  morphism $\varrho:Z_1\to Z_2$ is
said to be a Poisson morphism if
\be
\{f\circ \varrho,f'\circ\varrho\}_1=\{f,f'\}_2\circ\varrho, \qquad
f,f'\in C^\infty(Z_2),
\ee
or, equivalently, if $w_2=T\varrho \circ w_1$, where $T\varrho$ is
the tangent map to $\varrho$. Herewith, the rank of $w_1$ is
superior or equal to that of $w_2$. Therefore, there are no
pull-back and push-forward operations of Poisson structures in
general. Nevertheless, let us mention the following construction.

\begin{theorem}\label{p11.3} \mar{p11.3} Let $(Z,w)$ be a Poisson manifold and
$\pi:Z\to Y$ a fibration such that, for every pair of functions
$(f,g)$ on $Y$ and for each point $y\in Y$, the restriction of a
function $\{\pi^*f,\pi^*g\}$ to a fibre $\pi^{-1}(y)$ is constant,
i.e., $\{\pi^*f,\pi^*g\}$ is the pull-back onto $Z$ of some
function on $Y$. Then there exists a coinduced Poisson structure
$w'$ on $Y$ for which $\pi$ is a Poisson morphism.
\end{theorem}

\begin{example} \label{spr862} \mar{spr862}
The direct product $Z\times Z'$ of Poisson manifolds $(Z,w)$ and
$(Z',w')$ can be endowed with the product of Poisson structures,
given by a bivector field $w + w'$ such that the surjections
$\pr_1$ and $\pr_2$ are Poisson morphisms.
\end{example}

A function $f\in C^\infty(Z)$ is called the Casimir function of a
Poisson structure on $Z$ if its Poisson bracket with any function
on $Z$ vanishes. Casimir functions form a real ring $\cC(Z)$. In
particular, a symplectic manifold admits only constant Casimir
functions.

A vector field $u$ on a Poisson manifold $(Z,w)$ is an
infinitesimal generator of a local one-parameter group of Poisson
automorphisms iff the Lie derivative
\mar{gm508}\beq
\bL_uw=[u,w]_{\rm SN} \label{gm508}
\eeq
vanishes. It is called the canonical vector field for a Poisson
structure $w$. In particular, for any real smooth function $f$ on
a Poisson manifold $(Z,w)$, let us put
\mar{gm509}\beq
\vt_f=w^\sh (df)= -w\lfloor df=w^{\m\nu}\dr_\m f\dr_\nu.
\label{gm509}
\eeq
It is a canonical vector field, called the Hamiltonian vector
field of a function $f$ with respect to a Poisson structure $w$.
Hamiltonian vector fields fulfil the relations
\mar{m81}\beq
\{f,g\}=\vt_f\rfloor dg, \qquad [\vt_f,\vt_g]=\vt_{\{f,g\}},
\qquad f,g\in C^\infty(Z). \label{m81}
\eeq

For instance, the Hamiltonian vector field $\vt_f$ (\ref{z100}) of
a function $f$ on a symplectic manifold $(Z,\Om)$ coincides with
that (\ref{gm509}) with respect to the corresponding Poisson
structure $w_\Om$. The Poisson bracket defined by a symplectic
form $\Om$ reads
\be
\{f,g\}=\vt_g\rfloor\vt_f\rfloor\Om.
\ee

Since a Poisson manifold $(Z,w)$ is assumed to be regular, the
range $\bT=w^\sh(T^*Z)$ of the morphism (\ref{m51}) is a subbundle
of $TZ$ called the characteristic distribution on $(Z,w)$. It is
spanned by Hamiltonian vector fields, and it is involutive by
virtue of the relation (\ref{m81}). It follows that a Poisson
manifold $Z$ admits local adapted coordinates in Theorem
\ref{c11.0}. Moreover, one can choose particular adapted
coordinates which bring a Poisson structure into the following
canonical form.

\begin{theorem}\label{canpo} \mar{canpo} For any point $z$ of a $k$-dimensional Poisson
manifold $(Z,w)$, there exist coordinates
\mar{m111}\beq
(z^1,\dots,z^{k-2m},q^1,\dots,q^m,p_1,\dots,p_m) \label{m111}
\eeq
on a neighborhood of $z$ such that
\be
w=\frac{\dr}{\dr p_i}\w\frac{\dr}{\dr q^i},\qquad
\{f,g\}=\frac{\dr f}{\dr p_i} \frac{\dr g}{\dr q^i}
   - \frac{\dr f}{\dr q^i} \frac{\dr g}{\dr p_i}.
\ee
\end{theorem}

The coordinates (\ref{m111}) are called the canonical or Darboux
coordinates for the Poisson structure $w$. The Hamiltonian vector
field of a function $f$ written in this coordinates is
\be
\vt_f= \dr^if\dr_i - \dr_if\dr^i.
\ee
Of course, the canonical coordinates for a symplectic form $\Om$
in Theorem \ref{spr825} also are canonical coordinates in Theorem
\ref{canpo} for the corresponding non-degenerate Poisson bivector
field $w$, i.e.,
\be
\Om=dp_i\w dq^i, \qquad w=\dr^i\w \dr_i.
\ee
With respect to these coordinates, the mutually inverse bundle
isomorphisms $\Om^\fl$ (\ref{m52}) and $w^\sh$ (\ref{m51}) read
\be
&& \Om^\fl: v^i\dr_i + v_i\dr^i\to -v_idq^i+ v^idp_i, \\
&& w^\sh: v_idq^i+ v^idp_i \to v^i\dr_i- v_i\dr^i.
\ee

\subsection{Symplectic foliations}

Integral manifolds of the characteristic distribution $\bT$ of a
$k$-dimensional Poisson manifold $(Z,w)$ constitute a (regular)
foliation $\cF$ of $Z$ whose tangent bundle $T\cF$ is $\bT$. It is
called the characteristic foliation of a Poisson manifold. By the
very definition of the characteristic distribution $\bT=T\cF$, a
Poisson bivector field $w$ is subordinate to $\op\w^2 T\cF$.
Therefore, its restriction $w|_F$ to any leaf $F$ of $\cF$ is a
non-degenerate Poisson bivector field on $F$. It provides $F$ with
a non-degenerate Poisson structure $\{,\}_F$ and, consequently, a
symplectic structure. Clearly, the local Darboux coordinates for
the Poisson structure $w$ in Theorem \ref{canpo} also are the
local adapted coordinates
\be
(z^1,\ldots,z^{k-2m}, z^i=q^i,z^{m+i}=p_i), \qquad i=1,\ldots,m,
\ee
(\ref{spr850}) for the characteristic foliation $\cF$, and the
symplectic structures along its leaves read
\be
\Om_F=dp_i\w dq^i.
\ee

In particular, it follows that Casimir functions of a Poisson
structure are constant on leaves of its characteristic symplectic
foliation.

Since any foliation is locally simple, a local structure of an
arbitrary Poisson manifold reduces to the following
\cite{vais,wein}.

\begin{theorem}\label{wenst1} \mar{wenst1} Each point of a Poisson manifold has an
open neighborhood which is Poisson equivalent to the product of a
manifold with the zero Poisson structure and a symplectic
manifold.
\end{theorem}

Provided with this symplectic structure, the leaves of the
characteristic foliation of a Poisson manifold $Z$ are assembled
into a symplectic foliation of $Z$ as follows.

Let $\cF$ be an even dimensional foliation of a manifold $Z$. A
$\wt d$-closed non-degenerate leafwise two-form $\Om_\cF$ on a
foliated manifold $(Z,\cF)$ is called symplectic. Its pull-back
$i_F^*\Om_\cF$ onto each leaf $F$ of $\cF$ is a symplectic form on
$F$. A foliation $\cF$ provided with a symplectic leafwise form
$\Om_\cF$ is called the symplectic foliation.

If a symplectic leafwise form $\Om_\cF$ exists, it yields a bundle
isomorphism
\be
\Om_\cF^\fl: T\cF\op\to_Z T\cF^*, \qquad \Om_\cF^\fl:v\to -
v\rfloor\Om_\cF(z), \qquad v\in T_z\cF.
\ee
The inverse isomorphism $\Om_\cF^\sh$ determines a bivector field
\mar{spr904}\beq
w_\Om(\al,\bt)=\Om_\cF(\Om_\cF^\sh(i^*_\cF\al),\Om_\cF^\sh(i^*_\cF\bt)),
\qquad  \al,\bt\in T_z^*Z, \quad z\in Z, \label{spr904}
\eeq
on $Z$ subordinate to $\op\w^2T\cF$. It is a Poisson bivector
field. The corresponding Poisson bracket reads
\mar{spr902}\beq
\{f,f'\}_\cF=\vt_f\rfloor \wt df', \qquad \vt_f\rfloor\Om_\cF=-\wt
df, \qquad \vt_f=\Om_\cF^\sh(\wt df).\label{spr902}
\eeq
Its kernel is $S_\cF(Z)$.

Conversely, let $(Z,w)$ be a Poisson manifold and $\cF$ its
characteristic foliation. Since Ann$\,T\cF\subset T^*Z$ is
precisely the kernel of a Poisson bivector field $w$, a bundle
homomorphism
\be
w^\sh: T^*Z\op\to_Z TZ
\ee
factorizes in a unique fashion
\mar{lmp03}\beq
w^\sh: T^*Z\ar_Z^{i^*_\cF} T\cF^*\ar_Z^{w^\sh_\cF}
T\cF\ar_Z^{i_\cF} TZ \label{lmp03}
\eeq
through a bundle isomorphism
\mar{lmp02}\beq
w_\cF^\sh: T\cF^*\op\to_Z T\cF,  \qquad w^\sh_\cF:\al\to
-w(z)\lfloor \al, \qquad \al\in T_z\cF^*. \label{lmp02}
\eeq
The inverse isomorphism $w_\cF^\fl$ yields a symplectic leafwise
form
\mar{spr903}\beq
\Om_\cF(v,v')=w(w_\cF^\fl(v),w_\cF^\fl(v')), \qquad  v,v'\in
T_z\cF, \qquad z\in Z. \label{spr903}
\eeq
The formulas (\ref{spr904}) and (\ref{spr903}) establish the
equivalence between the Poisson structures on a manifold $Z$ and
its symplectic foliations.

\subsection{Group action on Poisson manifolds}

Turn now to a group action on Poisson manifolds. By $G$ throughout
is meant a real connected Lie group, $\cG$ is its right Lie
algebra, and $\cG^*$ is the Lie coalgebra (see Section 7.5).

We start with the symplectic case. Let a Lie group $G$ act on a
symplectic manifold $(Z,\Om)$ on the left by symplectomorphisms.
Such an action of $G$ is called symplectic. Since $G$ is
connected, its action on a manifold $Z$ is symplectic iff the
homomorphism $\ve\to\xi_\ve$, $\ve\in\cG$, (\ref{spr941}) of a Lie
algebra $\cG$ to a Lie algebra $\cT_1(Z)$ of vector fields on $Z$
is carried out by canonical vector fields for a symplectic form
$\Om$ on $Z$. If all these vector fields are Hamiltonian, an
action of $G$ on $Z$ is called a Hamiltonian action. One can show
that, in this case, $\xi_\ve$, $\ve\in\cG$, are Hamiltonian vector
fields of functions on $Z$ of the following particular type.

\begin{proposition}\label{gen4}  \mar{gena4} An action of a Lie group $G$ on
a symplectic manifold $Z$ is Hamiltonian iff there exists a
mapping
\mar{spr943}\beq
\wh J: Z\to \cG^*, \label{spr943}
\eeq
called the momentum mapping,  such that
\mar{z210}\beq
\xi_\ve\rfloor\Om=- dJ_\ve, \qquad J_\ve(z)=\lng\wh J(z),\ve\rng,
\qquad  \ve\in \cG. \label{z210}
\eeq
\end{proposition}

The momentum mapping (\ref{spr943}) is defined up to a constant
map. Indeed, if $\wh J$ and $\wh J'$ are different momentum
mappings for the same symplectic action of $G$ on $Z$, then
\be
d(\lng \wh J(z)-\wh J'(z),\ve\rng)=0, \qquad  \ve\in\cG.
\ee
Given $g\in G$, let us us consider the difference
\mar{z220}\beq
\si(g) =\wh J(gz)-{\rm Ad}^*g(\wh J(z)), \label{z220}
\eeq
where ${\rm Ad}^*g$ is the coadjoint representation (\ref{z211})
on $\cG^*$. One can show that the difference (\ref{z220}) is
constant on a symplectic manifold $Z$ \cite{abr}. A momentum
mapping $\wh J$ is called equivariant if $\si(g)=0$, $g\in G$.

\begin{example}\label{imp} \mar{imp}
Let a symplectic form on $Z$ be exact, i.e., $\Om=d\thh$, and let
$\thh$ be $G$-invariant, i.e.,
\be
\bL_{\xi_\ve}\thh=d(\xi_\ve\rfloor\thh) +\xi_\ve\rfloor\Om=0,
\qquad
  \ve\in\cG.
\ee
Then the momentum mapping $\wh J$ (\ref{spr943}) can be given by
the relation
\be
\lng\wh J(z),\ve\rng=(\xi_\ve\rfloor\thh)(z).
\ee
It is equivariant. In accordance with the relation (\ref{z211}),
it suffices to show that
\be
J_\ve(gz)=J_{{\rm Ad}\,g^{-1}(\ve)}(z), \qquad (\xi_\ve\rfloor
\thh)(gz)=(\xi_{{\rm Ad}\,g^{-1}(\ve)}\rfloor\thh)(z).
\ee
This holds by virtue of the relation (\ref{z212}). For instance,
let $T^*Q$ be a symplectic manifold equipped with the canonical
symplectic form $\Om_T$ (\ref{m83}). Let a left action of a Lie
group $G$ on $Q$ have the infinitesimal generators
$\tau_m=\ve^i_m(q)\dr_i$. The canonical lift of this action onto
$T^*Q$ has the  infinitesimal generators
\mar{z216}\beq
\xi_m=\wt\tau_m=ve^i_m\dr_i -p_j\dr_i\ve_m^j\dr^i, \label{z216}
\eeq
and preserves the canonical Liouville form $\Xi$ on $T^*Q$. The
$\xi_m$ (\ref{z216}) are Hamiltonian vector fields of the
functions $J_m=\ve_m^i(q)p_i$, determined by the equivariant
momentum mapping $\wh J=\ve^i_m(q)p_i\ve^m$.
\end{example}

\begin{theorem}
A momentum mapping $\wh J$ associated to a symplectic action of a
Lie group $G$ on a symplectic manifold $Z$ obeys the relation
\mar{z223}\beq
\{J_\ve,J_{\ve'}\}= J_{[\ve,\ve']} - \lng T_e\si(\ve'),\ve\rng.
\label{z223}
\eeq
\end{theorem}

In the case of an equivariant momentum mapping, the relation
(\ref{z223}) leads to a homomorphism
\mar{z224}\beq
\{J_\ve,J_{\ve'}\}= J_{[\ve,\ve']} \label{z224}
\eeq
of a Lie algebra $\cG$ to a Poisson algebra of smooth functions on
a symplectic manifold $Z$ (cf. Proposition \ref{spr948} below).

Now let a Lie group $G$ act on a Poisson manifold $(Z,w)$ on the
left by Poisson automorphism. This is a Poisson action. Since $G$
is connected, its action on a manifold $Z$ is a Poisson action iff
the homomorphism $\ve\to\xi_\ve$, $\ve\in\cG$, (\ref{spr941}) of a
Lie algebra $\cG$ to a Lie algebra $\cT_1(Z)$ of vector fields on
$Z$ is carried out by canonical vector fields for a Poisson
bivector field $w$, i.e., the condition (\ref{gm508}) holds. The
equivalent conditions are
\be
&&\xi_\ve(\{f,g\})=\{\xi_\ve(f),g\} +\{f, \xi_\ve(g)\}, \qquad f,g\in
C^\infty(Z), \\
&&\xi_\ve(\{f,g\})=[\xi_\ve,\vt_f](g)- [\xi_\ve,\vt_g](f),\\
&& [\xi_\ve,\vt_f] =\vt_{\xi_\ve(f)},
\ee
where $\vt_f$ is the Hamiltonian vector field (\ref{gm509}) of a
function $f$.

A Hamiltonian action of $G$ on a Poisson manifold $Z$ is defined
similarly to that on a symplectic manifold. Its infinitesimal
generators are tangent to leaves of the symplectic foliation of
$Z$, and there is a Hamiltonian action of $G$ on every symplectic
leaf. Proposition \ref{gen4} together with the notions of a
momentum mapping and an equivariant momentum mapping also are
extended to a Poisson action. However, the difference $\si$
(\ref{z220}) is constant only on leaves of the symplectic
foliation of $Z$ in general. At the same time, one can say
something more on an equivariant momentum mapping (that also is
valid for a symplectic action).

\begin{proposition} \label{spr948} \mar{spr948}
An equivariant momentum mapping $\wh J$  (\ref{spr943}) is a
Poisson morphism to the Lie coalgebra $\cG^*$, provided with the
Lie--Poisson structure (\ref{gm510}).
\end{proposition}

\end{document}